%% file: main.tex
\documentclass[11pt]{article}
\usepackage{lmodern} 
\usepackage[T1]{fontenc}

\usepackage{color}
\usepackage[colorlinks=true,linkcolor=blue,citecolor=blue]{hyperref}
\usepackage{amsmath, amssymb, amsthm}
\usepackage{mathtools}
\usepackage[margin=1in]{geometry}
\usepackage{graphics}
\usepackage{pifont}
\usepackage{tikz}
\usepackage{bbm, bm}

\usetikzlibrary{arrows.meta}
\usepackage{environ}
\usepackage{framed}
\usepackage{url}
\usepackage[linesnumbered,ruled,vlined]{algorithm2e}
\usepackage[noend]{algpseudocode}
\usepackage[labelfont=bf]{caption}
\usepackage{framed}
\usepackage[framemethod=tikz]{mdframed}
\usepackage{appendix}
\usepackage{graphicx}
\usepackage[textsize=tiny]{todonotes}
\usepackage{tcolorbox}
\allowdisplaybreaks[3]
\usepackage{nicefrac}
\usepackage{thm-restate}
\usepackage[noabbrev,capitalize,nameinlink]{cleveref}
\crefname{equation}{}{}
\usepackage{pdfsync}
\usepackage{tabularx, ragged2e, booktabs}

\usepackage{lipsum}
\usepackage{enumerate}
\usepackage{float}

\usepackage{subcaption}

\DontPrintSemicolon
\SetKw{KwAnd}{and}
\SetProcNameSty{textsc}
\SetFuncSty{textsc}

\newcommand\remove[1]{}
\renewcommand\paragraph[1]{\smallskip \noindent \textbf{#1}}

\newtheorem{lemma}{Lemma}[section]
\newtheorem*{lemma*}{Lemma}
\newtheorem{theorem}[lemma]{Theorem}

\newtheorem*{corollary*}{Corollary}

\newtheorem*{theorem*}{Theorem}
\newtheorem*{inducthyp*}{Inductive Hypothesis}
\newtheorem*{definition*}{Definition}

\newtheorem{prob}[lemma]{Problem}
\newtheorem*{rem*}{Remark}

\newcommand\R{\mathbb{R}}

\newcommand{\eps}{\epsilon}
\renewcommand{\O}{\widetilde{O}}

\renewcommand{\l}{\langle}
\renewcommand{\r}{\rangle}

\newcommand{\otilde}{\O}

\newcommand{\supp}{\mathrm{supp}}

\newcommand{\bc}{\bm{c}}
\newcommand{\bd}{\boldsymbol{d}}
\renewcommand{\bf}{\bm{f}}
\newcommand{\bg}{\boldsymbol{g}}

\newcommand{\br}{\boldsymbol{r}}
\newcommand{\bw}{\boldsymbol{w}}

\newcommand{\bu}{\boldsymbol{u}}

\newcommand{\bx}{\boldsymbol{x}}
\newcommand{\by}{\boldsymbol{y}}
\newcommand{\bell}{\boldsymbol{\ell}}

\newcommand{\bDelta}{\boldsymbol{\Delta}}

\newcommand{\diag}{\mathrm{diag}}

\newcommand{\IncrPNorm}{\textsc{IncrementalPNorm}}
\newcommand{\IncrMWU}{\textsc{IncrementalMWU}}

\renewcommand{\hat}{\widehat}
\renewcommand{\tilde}{\widetilde}

\DeclareFontFamily{U}{mathb}{\hyphenchar\font45}
\DeclareFontShape{U}{mathb}{m}{n}{<5> <6> <7> <8> <9> <10> gen * mathb
<10.95> mathb10 <12> <14.4> <17.28> <20.74> <24.88> mathb12}{}
\DeclareSymbolFont{mathb}{U}{mathb}{m}{n}
\DeclareMathSymbol{\rcirclearrow}{\mathbin}{mathb}{'367}

\newcommand{\wt}{\widetilde}
\newcommand{\wh}{\widehat}

\foreach \x in {A,...,Z}{%
	\expandafter\xdef\csname m\x\endcsname{\noexpand\mathbf{\x}}
}

\foreach \x in {A,...,Z}{%
	\expandafter\xdef\csname c\x\endcsname{\noexpand\mathcal{\x}}
}

\newif\ifrandom
\randomtrue

\newcommand{\defeq}{\stackrel{\mathrm{\scriptscriptstyle def}}{=}}

\newcommand{\poly}{{\mathrm{poly}}}

\newcommand{\bb}{\boldsymbol{b}}
\newcommand{\ba}{\boldsymbol{a}}

\renewcommand{\l}{\langle}
\renewcommand{\r}{\rangle}

\newcommand{\norm}[1]{\left\lVert#1\right\rVert}
\newcommand{\Abs}[1]{\left|#1\right|}

\newcommand{\todolater}[1]{}

\DeclareUnicodeCharacter{2113}{$\ell$}

\usepackage[backend=biber, isbn=false, style=alphabetic, backref=true,
doi=false, url=false, maxcitenames=10, mincitenames=5,
maxalphanames=10, maxbibnames=10, minbibnames=5, minalphanames=3,
defernumbers=true, sortlocale=en_US, sorting=ynt, sortcites]{biblatex}
\AtBeginBibliography{\small}
\begin{document}
\pagenumbering{gobble}
\title{Incremental Approximate Maximum Flow on\\
Undirected Graphs in Subpolynomial Update Time}

\author{
Jan van den Brand\\ Georgia Tech\\ vdbrand@gatech.edu
\and
Li Chen\thanks{Li Chen was supported by NSF Grant CCF-2106444.}\\ Georgia Tech\\ lichen@gatech.edu
\and
Rasmus Kyng\thanks{The research leading to these results has received funding from the grant ``Algorithms and complexity for high-accuracy flows and convex optimization'' (no. 200021 204787) of the Swiss National Science Foundation.}\\ ETH Zurich \\ kyng@inf.ethz.ch 
\and
Yang P. Liu\thanks{Yang P. Liu was supported by NSF CAREER Award CCF-1844855, NSF Grant CCF-1955039, and a Google Research Fellowship.} \\ Stanford University \\ yangpliu@stanford.edu
\and
Richard Peng\thanks{Richard Peng was partially supported by the Natural Sciences and Engineering Research Council of Canada (NSERC) Discovery Grant RGPIN-2022-03207} \\ University of Waterloo \\
y5peng@uwaterloo.ca
\and
Maximilian Probst Gutenberg\footnotemark[2]\\ ETH Zurich\\ maxprobst@ethz.ch
\and
Sushant Sachdeva\thanks{Sushant Sachdeva was supported by an NSERC Discovery Grant RGPIN-2018-06398, an Ontario Early Researcher Award (ERA) ER21-16-284, and a Sloan Research Fellowship.
} \\ University of Toronto \\ sachdeva@cs.toronto.edu
\and 
Aaron Sidford\thanks{Aaron Sidford was supported by a Microsoft Research Faculty Fellowship, NSF CAREER Award CCF-1844855, NSF Grant CCF-1955039, a PayPal research award, and a Sloan Research Fellowship.} \\
Stanford University\\
sidford@stanford.edu
}
\maketitle

\begin{abstract}
\input{abstract.tex}
\end{abstract}

\newpage
\setcounter{tocdepth}{2}
\tableofcontents

\normalsize
\pagebreak
\pagenumbering{arabic}

\input{intro.tex}

\input{prelim.tex}

\input{overview.tex}

\input{refinement.tex}

\input{minratio.tex}

\subsection*{Acknowledgements}

We thank the anonymous reviewers for their feedback and suggestions.

\begin{refcontext}[sorting=nyt]
\printbibliography[heading=bibintoc]
\end{refcontext}

\end{document}

%% file: abstract.tex
We provide an algorithm which, with high probability, maintains a $(1-\eps)$-approximate maximum flow on an undirected graph undergoing $m$-edge additions in amortized $m^{o(1)} \eps^{-3}$ time per update. To obtain this result, we provide a more general algorithm that solves what we call the \emph{incremental, thresholded, $p$-norm flow problem} that asks to determine the first edge-insertion in an undirected graph that causes the minimum $\ell_p$-norm flow to decrease below a given threshold in value. Since we solve this thresholded problem, our data structure succeeds against an adaptive adversary that can only see the data structure's output. Furthermore, since our algorithm holds for $p = 2$, we obtain improved algorithms for dynamically maintaining the effective resistance between a pair of vertices in an undirected graph undergoing edge insertions.

Our algorithm builds upon previous dynamic algorithms for approximately solving the minimum-ratio cycle problem that underlie previous advances on the maximum flow problem [Chen-Kyng-Liu-Peng-Probst Gutenberg-Sachdeva, FOCS '22] as well as recent dynamic maximum flow algorithms [v.d.Brand-Liu-Sidford, STOC '23]. Instead of using interior point methods, which were a key component of these recent advances, our algorithm uses an optimization method based on $\ell_p$-norm iterative refinement and the multiplicative weight update method. This ensures a monotonicity property in the minimum-ratio cycle subproblems that allows us to apply known data structures and bypass issues arising from adaptive queries.

%% file: intro.tex
\section{Introduction}

The design and analysis of dynamic graph algorithms is a rich, well-studied research area.
Researchers have studied dynamic variations of fundamental graph problems such as minimum spanning tree, shortest path,
and bipartite matching.
Recent progress on dynamic matching has been widely celebrated \cite{B23,BKSW23}, and dynamic graph algorithms play a key role in recent advances in static graph algorithms, 
yielding the first nearly-linear time algorithms for bipartite matching and maximum flow (maxflow) in dense graphs \cite{BLNPSSSW20,BLLSSSW21} and almost-linear time algorithm for maxflow \cite{CKLPPS22}.
(See \Cref{sec:related}, for additional related work.)

Despite the many successes of dynamic graph algorithms, \emph{dynamic maxflow},
another class of core problems, has seen relatively little progress. This is perhaps due in part to strong conditional hardness results. In the incremental and decremental settings, where, respectively, edges are only added or only removed,
exactly maintaining the $s$-$t$ flow value requires $\Omega(n)$ time per update, even in directed unit capacity graphs, assuming the Online-Matrix-vector (OMv)-conjecture.
This was shown by \cite{Dahlgaard16}, building on earlier work by \cite{HenzingerKNS15} that introduced the OMv-conjecture and showed a $O(\sqrt{m})$ time per update lower bound for the same problems.

Moving to approximate solutions opens the possibility of faster algorithms.
Recently, \cite{GH23} gave an
$(1-\epsilon)$-approximation algorithm for incremental unit capacity maxflow with amortized update time $\hat{O}(\epsilon^{-1/2} \sqrt{m})$\footnote{In this paper, we use $\wh{O}(\cdot)$ to suppress subpolynomial $n^{o(1)}$ factors.}, an improvement for sparse graphs in the low accuracy regime. Additionally, \cite{BLS23} recently
gave an $(1-\epsilon)$-approximation algorithm for incremental maxflow and minimum-cost flow\footnote{In both cases, this is assuming polynomially bounded capacities; we make this assumption throughout the paper.} with amortized update time $\hat{O}(\epsilon^{-1} \sqrt{n})$ by dynamizing and building upon the recent almost-linear time algorithm for maxflow of \cite{CKLPPS22}.
For constant $\epsilon$, this runtime is faster than the conditional lower bound for exact incremental maxflow.

In this work, we ask whether \emph{we can build upon this recent progress and give a subpolynomial amortized update time algorithm for the dynamic maxflow problem?} We answer this in the affirmative by developing such an algorithm for approximate incremental undirected maxflow.

\paragraph{Dynamizing Static Maxflow.} 
To obtain our result, we build upon the recent advance of \cite{BLS23} and the work that underlies it, as well as distinct lines of research related to approximate undirected maxflow.
 \cite{BLS23} essentially dynamizes the almost linear-time algorithm for maxflow \cite{CKLPPS22}.
\cite{CKLPPS22} came at the end of a long line of research that focused on solving flow problems by combining graph theoretic tools with interior point methods (IPMs), a class of continuous optimization methods which obtain high-accuracy solutions to convex optimization problems
\cite{DS08,M13,LS14,M16,CMSV17,LS20,KLS20,AMV20,BLNPSSSW20,BLLSSSW21,GLP21,AMV21,DGGLPSY22,BGJLLPS22,CKLPPS22}.
The \cite{CKLPPS22} IPM relies on solving an $\ell_1$ flow update problem known as ``(undirected) min-ratio cycle.''\footnote{In this paper we refer to this problem as \emph{min-ratio cycle}, omitting the term ``undirected.'' 
Elsewhere in the literature, but never in this paper, \emph{min-ratio cycle} may refer to a variant with edge-direction constraints.
}
This problem is solved $m^{1+o(1)}$ times using a data structure with amortized $m^{o(1)}$ time per update.
\cite{BLS23}, showed how to dynamize this IPM, but obtained an $n^{1/2+o(1)}$ amortized time per update due to the cost of adapting the min-ratio cycle data structures to this setting.

There are two key ideas in \cite{BLS23}. The first is that the $\ell_1$-IPM of \cite{CKLPPS22} can be naturally extended to the incremental setting. 
In particular, they showed how the IPM can be used to solve a threshold variant of the incremental maximum flow problem, i.e., detecting the first update that causes the maximum flow value to increase above a threshold. However, this extension creates a challenge: The \cite{CKLPPS22} data structure for min-ratio cycle does \emph{not} work against an adaptive adversary. Instead, \cite{CKLPPS22} crucially leverages stability properties of their $\ell_1$ IPM that ultimately determines the update problems, and uses these stability properties to guarantee that while their data structure may occasionally fail, this will occur infrequently.
In the incremental setting, \cite{BLS23} cannot leverage this guarantee to ensure their data structure works.
This leads to the second central idea of \cite{BLS23}.
They develop a new version of the \cite{CKLPPS22} data structure that works against adaptive adversaries, at the expense of increasing the amortized time per update from $m^{o(1)}$ to $n^{1/2+o(1)}$.

At a high level, our approach is motivated by \cite{BLS23}, but ultimately we develop a fundamentally different optimization approach, which yields more tractable update problems. A key observation enabling our algorithms, is that the min-ratio cycle problem succeeds against adaptive adversaries provided that there is a particular monotonicity in the updates. We change the optimization framework to yield subproblems which can be solved by such monotonic updates. Because of this, we manage to show that in our setting, the data structure of \cite{CKLPPS22} can directly solve the update problems in the incremental setting, with only $m^{o(1)}$ amortized time per update.

It is worth mentioning that the authors of this work recently gave a deterministic min-cost flow algorithm \cite{detMaxFlow}. However, \cite{detMaxFlow} still uses a version of the data structure of \cite{CKLPPS22} that does not work against adaptive adversaries, and critically still uses stability properties of the update sequence to argue correctness.

\paragraph{$\ell_p$-norm Flow and Approximate Undirected Maxflow.}
To leverage that the min-ratio cycle data structure succeeds against monotonic adversaries, we turn to an approach motivated by lines of research on 
algorithms for (static) $\ell_p$-norm flow and approximate undirected maxflow.
One important line of research yielded undirected approximate maxflow in  $\tilde{O}(\epsilon^{-1}m)$ time~\cite{CKMST11, S13, KLOS14, P16, S17}.
This sequence of works combined first-order continuous optimization methods with graph theoretic tools.
A second line of work focused on a broader class of flow problems, namely $\ell_p$-norm flows~\cite{AKPS19, AS20, APS19, KPSW19, ABKS21},
and developed iterative optimization methods tailored to $\ell_p$-norm objectives.
The problem of $\ell_p$-norm flows asks for a flow $\bf$ that routes a given demand and minimizes its $\ell_p$-norm, $\|\bf\|_p.$
This problem is maxflow for $p = \infty$ and setting $p = O(\eps^{-1} \log m)$ yields $(1-\eps)$-approximate undirected maxflow.

\paragraph{Our Approach to Incremental Flow Problems.}  We show that the desired monotonicity properties for the subproblems can be achieved in the setting of $\ell_p$-norm flows by adapting the optimization framework. Even though $\ell_p$-norm flow is less general than directed maxflow, it has interesting consequences including approximate undirected maxflow and effective resistances~\cite{CKMST11,SS08}.

Switching to $\ell_p$-norm flows allows us to use the $\ell_p$ iterative refinement framework developed in \cite{AKPS19, AS20, APS19}, and study the smoothed $\ell_p$-norm flow problems of \cite{KPSW19}.
The iterative refinement framework shows that smoothed $\ell_p$-norm flow computation can be accomplished by a small number of iterations of a \emph{refinement step}. In our context $\hat{O}(p)$ steps suffice.

To solve each refinement step problem, we use a $\ell_1$ multiplicative weight update method (MWU). 
Our method lets us solve smoothed $\ell_p$-norm flow to $m^{o(1)}$ accuracy by solving a min-ratio cycle problem roughly $m^{1+o(1)}$ times.
Crucially, we show that our MWU induces a \emph{monotonicity} property in our min-ratio cycle problems.
Concretely, our sequence of approximate min-ratio cycle problems only change by (1) edge insertions and (2) edge lengths increases (see \cref{prob:monoMRC}).
In this way, we create a more tractable data structure problem than those in \cite{BLS23}, and we show that this problem can be solved using data structures from \cite{CKLPPS22} with $m^{o(1)}$ amortized update time.

Finally, combining our monotonic $\ell_1$-MWU for computing refinement steps with iterative refinement, we obtain an incremental algorithm for (a decision version of) smoothed $\ell_p$-norm flows.
From this, we derive an algorithm for approximate incremental undirected maxflow and for incremental electrical flow.
Using the incremental electrical flow algorithm, for a fixed pair of vertices $s,t$, we can detect in an incremental graph the first time the effective resistance between $s$ and $t$ drops below a given threshold. 

If we were only focusing on approximate incremental maxflow instead of the more general $\ell_p$-norm flows,
a similar monotonic data structure problem could also be obtained by using an $\ell_1$-oracle MWU to compute each update step of a first-order optimization method such as one used by \cite{S13, KLOS14, P16, S17}.
Wrapping this inside a first-order $\ell_{\infty}$ optimization approach would then yield a similar algorithm for approximate incremental maxflow.

\input{results}

\subsection{Additional Related Work}
\label{sec:related}

\paragraph{Dynamic Flows and Matching.}
Exact dynamic maxflow on unit capacity graphs can be maintained in $O(m)$ time per update by performing one augmentation per update, or in $O(n)$ amortized time in the incremental setting
\cite{GK21,KumarG03}. 
On planar graphs, it can be maintained in the fully dynamic setting in $\O(n^{2/3})$ update and query time
\cite{ItalianoNSW11} and in the incremental setting with $\O(\sqrt{n})$ update and query time \cite{das2022near}.
Beyond that, there is more recent work on the approximate setting, with $(1-\epsilon)$-approximate incremental algorithms \cite{GH23,BLS23}, as discussed earlier.
Finally, in the fully dynamic setting, algorithms with super-constant (i.e.~polylog or subpolynomial) approximation ratios
and sublinear amortized update time \cite{CGHPS20} and for uncapacitated graphs with subpolynomial worst-case update time \cite{GRST21} are known.

Dynamic bipartite matching (which is a special case of directed maxflow) has also received significant attention in the approximate setting \cite{GuptaP13,Gupta14,BosekLSZ14,BernsteinS15,BernsteinS16,BaswanaGS15,PelegS16,Solomon16,BhattacharyaHN16,ArarCCSW18,CharikarS18,BernsteinFH21,BernsteinHR19,BhattacharyaK19,BehnezhadDHSS19,ChechikZ19,Wajc20,bernstein2020deterministic,BehnezhadLM20,BhattacharyaK21,Kiss21,BehnezhadK22,LeMSW22,GrandoniSSU22,RoghaniSW22,AssadiBKL22,BKSW23,BKS23,BlikstadK23}.
In the $(1-\epsilon)$-approximate regime, 
the current state-of-the-art is $\poly(1/\eps)$ update time for the incremental setting \cite{BlikstadK23}, 
$\poly(\log(n)/\epsilon)$ for the decremental setting \cite{BKS23mwu}, 
and $O(\sqrt{m}^{1-\Omega_{\eps}(1)})$ for the fully dynamic setting~\cite{BKS23}. 
For fully dynamic exact bipartite matching, the fastest update time is $O(n^{1.406})$ \cite{Sankowski07,BrandNS19}.

Finally, it is worth mentioning that recent works have leveraged numerical optimization methods based on entropy-regularized optimal transport and MWU to design dynamic algorithms for partially dynamic bipartite matching and positive linear programs \cite{JJST22,BKS23mwu}.

\paragraph{Edge Connectivity.}
The $k$-edge connectivity between two vertices $s,t$ can be seen as a maxflow of value up to $k$. 
Dynamic $k$-edge $st$-connectivity has been studied for small constant values of $k\le5$ \cite{GalilI91,GalilIb91,Frederickson91,WestbrookT92,DinitzV94,DinitzV95,HenzingerK97,EppsteinGIN97,DinitzW98,Thorup00,HolmLT01,HolmRT18}. For super-constant $k$, \cite{JinS21} presents a fully dynamic algorithm with $n^{o(1)}$ update time for $k=(\log n)^{o(1)}$.
\cite{ChalermsookDL+20} give an offline fully dynamic algorithm with $\widehat{O}(k^{O(k)})$ query time. These results all require small $k$ to be efficient, whereas our result has no such restrictions. We also point out the work in  \cite{thorup2007fully} which gives an algorithm to dynamically maintain the (value of) the global min-cut with $\tilde{O}(\sqrt{n})$ worst-case update time.

\paragraph{Dynamic Electric Flows.}
For $p=2$, our incremental $\ell_p$-norm flow algorithm can maintain a $(1-\epsilon)$-approximate electric flow between two fixed vertices $s,t\in V$ subject to edge insertions.
Dynamic electric flows have previously been studied in \cite{GoranciHP17a,DurfeeGGP19}. Such dynamic electric flow algorithms were also studied for the purpose of accelerating static maxflow and mincost flow algorithms \cite{GLP21,BGJLLPS22}.
The closely related concept of dynamic effective resistances has also been studied in the online dynamic \cite{CGHPS20} and offline dynamic setting \cite{LiPYZ20}.

\subsection{Paper Organization} 
In the remainder of the paper, we provide preliminaries in \cref{sec:prelim} and then give a more technical overview of our approach in \cref{sec:overview}.
We then we prove our main result, \cref{thm:incrPnormflow}, via iterative refinement in \cref{sec:incrIterRefine}. The incremental algorithm for the $\ell_p$-norm residual problem is then presented in \cref{sec:incrMWU}.
Finally, we argue that the approximate data structure \cite{CKLPPS22} can be used to implement the incremental multiplicative weight method in \cref{sec:monoMRC}.

%% file: results.tex
\subsection{Results}
\label{sec:results}

Here we present the main results of the paper. This section leverages a variety of notation, in particular graph theory conventions, all provided later in \cref{sec:prelim}.

A central result of this paper is an algorithm with subpolynomial update time for the following \emph{undirected incremental approximate maxflow problem }(which we abbreviate as \emph{incremental maxflow} in the remainder of the paper). Incremental maxflow is the dynamic data structure problem of maintaining a $(1+\epsilon)$-(multiplicative) approximate maximum flow in an undirected graph undergoing edge additions (hence the term incremental). 

\begin{restatable}[Undirected Incremental Approximate Maxflow Problem]{prob}{dynmaxflow}
\label{prob:dynmaxflow}
In the \emph{undirected incremental approximate maxflow problem (incremental maxflow)} 
we are given a finite set of $n$ vertices $V$, a distinct pair of elements $s$ and $t$, and a parameter $\epsilon > 0$. There are then a sequence of $m \le \poly(n)$ edge insertions where starting from $E = \emptyset$ an undirected edge $e$ is added to $E$ with integral capacity $\bu_e \in [1,U]$. The algorithm must maintain a $(1+\epsilon)$-approximate maximum flow in the capacitated graph $G = (V,E,\bu)$ before and after each edge addition. 
\end{restatable}

The main result of this paper is a randomized algorithm for the incremental maxflow problem that succeeds with high probability in $n$ (whp.) and implements each update in amortized $n^{o(1)} \epsilon^{-3}$ time. This is the first subpolynomial update time for an incremental maxflow problem which achieves even constant approximation.

Before stating our result, we comment on the adversary model. In our algorithms, we assume that the adversary can see the output flow, but not the internal randomness or information stored in the data structure. We call this an \emph{adaptive adversary}. We refer to the stronger adversary which can also see the internal randomness as a \emph{non-oblivious adversary}.

\begin{restatable}[Incremental Maxflow]{theorem}{incmflow}
\label{thm:incmaxflow}
There is an algorithm which solves incremental maxflow (\Cref{prob:dynmaxflow}) whp.\ in amortized $n^{o(1)} \epsilon^{-3}$ time per update against adaptive adversaries.
\end{restatable}

To obtain this result, we develop dynamic algorithms for the problem of computing smoothed $\ell_p$-norm flows \cite{KPSW19} for $p \geq 2$. The \emph{smoothed $\ell_p$-norm flow problem} asks to find a flow routing given vertex demands while minimizing a linear
plus quadratic plus $p$\textsuperscript{th} power objective on the flow. This problem generalizes both the popular and prevalent problems of computing electric flows (and therefore solving Laplacian systems)~\cite{ST04,KMP11,CKMPPRX14,JS21} as well as computing approximate maximum flows on undirected graphs~\cite{CKMST11, S13, KLOS14, P16, S17}.

\begin{restatable}[Smoothed $\ell_p$-norm flow]{prob}{pnormFlow}
\label{prob:pnormflow}
Given an undirected graph $G = (V, E)$, gradient vector $\bg^G \in \R^E$, edge resistances and weights $\br^G, \bw^G \in \R^E_+$, and demand vector $\bd \in \R^V$ the \emph{smoothed $\ell_p$-norm flow problem} asks to solve the following optimization problem
\begin{align}\label{eq:pnormflow}
    OPT = \min_{\mB^\top \bf = \bd} \cE(\bf) \enspace \text{ for } \enspace \cE(\bf) \defeq \l\bg^G, \bf\r + \norm{\mR^G \bf}_2^2 + \norm{\mW^G \bf}_p^p \,.
\end{align}
$f \in \R^E$ is said to be \emph{feasible} or \emph{routes the demands} if  $\mB^\top \bf = \bd$, $\cE(\bf)$ is called the \emph{energy} or \emph{smoothed objective value} of $\bf$, and a solution to \eqref{eq:pnormflow} is called a \emph{smoothed $\ell_p$-norm flow}.
\end{restatable}
For short, we refer to smoothed $\ell_p$-norm flows as simply \emph{$\ell_p$-norm flows} throughout.
Throughout we also assume that there is a feasible flow $\bf^{(0)}$ on the initial graph. This can be ensured by determining the first instance that $\bd$ is feasible using a simple union-data data structure.
This incurs only an additive $\O(1)$ cost in our data structures.

Our main technical result is a high-accuracy algorithm with subpolynomial update time for the following \emph{incremental thresholded $\ell_p$-norm flow} problem, or \emph{incremental $\ell_p$-norm flow} for short.
The incremental $\ell_p$-norm flow data structure detects the earliest moment when the optimal value to \eqref{eq:pnormflow} drops below a given threshold $F.$

\begin{restatable}[Incremental Thresholded $\ell_p$-Norm Flow]{prob}{incrPnormflow}
\label{prob:incrPnormflow}
Consider a dynamic instance of a $\ell_p$-Norm Flow $(G, \bg^G, \br^G, \bw^G, \bd)$ that is subject to edge insertion.
Let $m$ be the final number of edges in $G.$
Given an objective threshold $F \in \R$, and an error parameter $\eps > 0$, the problem of \emph{Incremental Threshold $\ell_p$-Norm Flow} asks, after the initialization or each edge insertion, to either
\begin{enumerate}
\item correctly certify that $OPT > F$, or
\item output a feasible flow $\bf$ with $\cE(\bf) \le F + \eps.$
\end{enumerate}
\end{restatable}

Combining an $\ell_1$-MWU with the dynamic min-ratio data structure of \cite{CKLPPS22}, we can solve \cref{prob:incrPnormflow} in almost linear time whp.\ against an adaptive adversary.

\begin{restatable}{theorem}{thmIncrPnormflow}\label{thm:incrPnormflow}
There is a randomized algorithm for \cref{prob:incrPnormflow} that given an initial flow $\bf^{(0)}$ and inputs $\eps,\bg^G, \br^G, \bw^G$ and $\bd$ with sizes bounded by $\O(1)$ in fixed point arithmetic., runs in $p^2 m^{1+o(1)} \log(\frac{\cE(\bf^{(0)}) - F}{\eps})$ time and succeeds whp.\ against adaptive adversaries.
\end{restatable}

As a result, taking $p=2$ yields an algorithm for $(1+\eps)$-approximate incremental electrical flow with subpolynomial update time.
This is the first subpolynomial time algorithm for constant accuracy incremental electrical flows.

We show how \cref{thm:incrPnormflow} leads to the incremental maxflow algorithm.
\begin{proof}[Proof of \Cref{thm:incmaxflow}]
The algorithm proceeds in about $\O(\eps^{-1})$ phases. In each phase, we determine when the congestion of the optimal flow has decreased by at least a $(1-\eps)$ factor. At the start of such a phase, we find the optimal maxflow $\bf$ using the almost-linear time algorithm of \cite{CKLPPS22}.
Let $C$ be the congestion of $\bf$.
We wish to determine when the optimal congestion is smaller than $e^{-\eps}C.$
From the previous discussion, any flow of congestion less than $e^{-\eps}C$ must have its $\ell_p$-norm at most $m (e^{-\eps}C)^p.$
Our flow $\bf$, on the other hand, has $p$-norm at most $mC^p.$
So we apply \cref{thm:incrPnormflow} with threshold $F = m (e^{-\eps}C)^p$ and error $m (e^{-\eps}C)^p$ as well.
When the data structure certifies that every flow has $\ell_p$-norm at least $m (e^{-\eps}C)^p$, we know that $\bf$ is still a $(1-\eps)$-approximate maxflow.
When the data structure outputs a new feasible flow $\bf'$ s.t. $\cE(\bf') \le F + m (e^{-\eps}C)^p = 2m (e^{-\eps}C)^p$, we know the congestion of $\bf'$ is at most $(2m (e^{-\eps}C)^p)^{1/p} \le e^{-\eps/2} C.$
This means that the optimal congestion drops at least by a factor of $(1+\eps/2)$ and we restart the whole algorithm with a newly computed maximum flow.
It restarts at most $\O(1/\eps)$ times because the congestion is between $[\exp(-\O(1)), \exp(\O(1))].$
\end{proof}

\paragraph{Bit Complexity.}
The number of exact arithmetic operations performed in the algorithm for \cref{thm:incrPnormflow} is only $p m^{1+o(1)} \log(\frac{\cE(\bf^{(0)}) - F}{\eps})$.
The additional $p$ dependency arises as our algorithms manipulate numbers of size $\O(p)$ due to the $p$-th power taken in the objective \eqref{eq:pnormflow} and we are using fixed point arithmetic.
A potential way to improve the dependency on $p$ is via floating point arithmetic and crude approximations.
That is, during the course of the algorithm for \cref{thm:incrPnormflow}, we only use $(1+\exp(-\O(1)))$-multiplicative approximation to the numbers we encountered and use an additional $\O(p)$ bits to represent their scales.
This way, the numbers encountered would be of size $\O(1)$ and we could shave one factor of $p$ from the runtime in \cref{thm:incrPnormflow}.
This could also improve the amortized update time for incremental approximate maximum flow to $n^{o(1)} \eps^{-2}.$

%% file: prelim.tex
\section{Preliminaries}
\label{sec:prelim}

In this section, we introduce notations we use throughout the paper.

\paragraph{General notation.} We denote vectors by boldface lowercase letters. We use uppercase boldface to denote matrices. Often, we use uppercase matrices to denote the diagonal matrices corresponding to lowercase vectors, such as $\mL = \diag(\bell)$. For vectors $\bx, \by$ we define the vector $\bx \circ \by$ as the entrywise product, i.e., $(\bx \circ \by)_i = \bx_i\by_i$. We also define the entrywise absolute value of a vector $|\bx|$ as $|\bx|_i = |\bx_i|$. We use $\l \cdot, \cdot \r$ as the vector inner product: $\l \bx, \by \r = \bx^\top\by = \sum_i \bx_i\by_i$. We elect to use this notation when $\bx, \by$ have superscripts (such as time indices) to avoid cluttering. For positive real numbers $a, b$ we write $a \approx_{\alpha} b$ for some $\alpha > 1$ if $\alpha^{-1}b \le a \le \alpha b$. 
For positive vectors $\bx, \by \in \R^{[n]}_+$, we say $\bx \approx_{\alpha} \by$ if $\bx_i \approx_{\alpha} \by_i$ for all $i \in [n]$.

\paragraph{Graphs.} In this paper, we consider multi-graphs $G$ with edge set $E(G)$ and vertex set $V(G)$. When the graph is clear from context, we use the shorthands $E$ for $E(G)$, $V$ for $V(G)$, $m = |E|$, and $n = |V|$. We assume that each edge $e \in E$ has an implicit direction, used to define its edge-vertex incidence matrix $\mB \in \R^{E \times V}$, i.e., $\mB_{e, u} = -1$, $\mB_{e, v}=+1$, and zero elsewhere for the row corresponding to the edge $e = (u, v).$
Abusing notation slightly, we often write $e = (u,v) \in E$ where $e$ is an edge in $E$ and $u$ and $v$ are the tail and head of $e$ respectively (note that technically multi-graphs do not allow for edges to be specified by their endpoints).

A vector $\bd \in \R^V$ is a demand vector if it is orthogonal to the all-ones vector, i.e., $\sum_{v \in V} \bd_v = 0$.
We say a flow $\bf \in \R^E$ routes a demand $\bd \in \R^V$ if $\mB^\top\bf=\bd$.
We say a flow $\bf$ is a circulation if it routes an all-zeros demand, i.e., each vertex has zero net flow.
For an edge $e = (u, v) \in G$ we let $\bb_e \in \R^V$ denote the demand vector of routing one unit from $u$ to $v$.

\paragraph{Dynamic Algorithms.}
We say $G$ is a \emph{dynamic} graph, if it undergoes \emph{batches} $U^{(1)}, U^{(2)}, \ldots$ of updates consisting of edge insertions/deletions that are applied to $G$. We use $|U^{(t)}|$ to denote the number of updates contained in the batch $U^{(t)}$.
The results on dynamic graphs in this article often only consider a subset of the update types and we therefore often state for each dynamic graph which updates are allowed. We say that a dynamic graph $G$ is incremental (and decremental) if it only undergoes edge insertions (and edge deletions respectively). Additionally, we say that the graph $G$, after applying the first $t$ update batches $U^{(1)}, U^{(2)}, \ldots, U^{(t)}$, is at \emph{stage} $t$ and denote the graph at this stage by $G^{(t)}$. Additionally, when $G$ is clear, we often denote the value of a variable $x$ at the end of stage $t$ of $G$ by $x^{(t)}$, or a vector $\bx$ at the end of stage $t$ of $G$ by $\bx^{(t)}$.

%% file: overview.tex
\section{Technical Overview}
\label{sec:overview}

Here we provide a technical overview of the approach we take to obtain the results outlined in \Cref{sec:results}. In \Cref{sec:overview:tools}, we briefly review a variety of previous tools which we leverage and obstacles that we overcome to obtain our results. In \Cref{sec:overview:adversary}, we then elaborate on our central insight about dynamic data structures for the minimum ratio cycle problem that fuels our results. In \Cref{sec:overview:iter}, we then discuss the dynamic optimization frameworks we use to leverage this data structure. Finally, in \Cref{sec:overview:altogether}, we discuss how we put these pieces together to obtain our results.

\subsection{Previous Tools and Obstacles.}
\label{sec:overview:tools}

We begin by describing the key tools of \cite{CKLPPS22} and \cite{BLS23} which solved maximum flow in almost linear time and obtained an $\wh{O}(\eps^{-1}\sqrt{n})$ time incremental algorithm for directed maximum flow. We elaborate on these tools and the obstacles for using them to achieve $n^{o(1)} \eps^{-O(1)}$ update time for undirected maximum flow.

\paragraph{Minimum Ratio Cycle.}
Both \cite{CKLPPS22} and \cite{BLS23} are built upon efficient data structures for approximately solving the min-ratio cycle problem on a fully-dynamic graph, i.e., finding a cycle that approximately minimizes
\begin{align}\label{eq:MRCOverview}
    \min_{\bc \in \R^E \text{ is a cycle}} \frac{\l\bg, \bc\r}{\norm{\mL \bc}_1}
\end{align}
where $\bg \in \R^E$ and $\bell \in \R^E_+$ are called \emph{edge gradients} and \emph{lengths} respectively.

Both data structures are based on similar ideas.
They maintain a $d$-level hierarchy of vertex and edge sparsification data structures.
For vertex sparsification, they use \emph{dynamic low stretch decompositions} which is studied in the context of dynamic shortest paths~\cite{CGHPS20}.
For edge sparsification, \cite{CKLPPS22} proposed a data structure that maintains graph spanners, which are sparse graphs that preserves all-pairs distances, under edge updates as well as vertex splits.
The original min-ratio cycle data structure has $m^{o(1)}$ update time and outputs $m^{o(1)}$-approximate solutions against oblivious adversaries. Additionally, the data structure succeeds against adaptive adversaries, as long as the inputs satisfy some additional ``stability properties''.
We elaborate on this later in \cref{sec:monoMRC}.
In \cite{BLS23}, they make the data structure adaptive at a higher update time of $n^{1/2+o(1)}.$

\paragraph{Interior Point Methods.}
The static algorithm of \cite{CKLPPS22} uses an IPM potential $\Phi(\bf)$ to find the maximum flow.
It starts at some initial flow where $\Phi(\bf) = \O(m)$ and iteratively makes progress until the potential is $\Phi(\bf) < -\O(m).$
At that point, an optimal flow is obtained by standard rounding techniques.
When the current flow is $\bf$, the algorithm finds a $m^{o(1)}$-approximate cycle $\bc$ to an instance of min-ratio cycle \eqref{eq:MRCOverview} with $\bg \approx \nabla \Phi(\bf)$ and $\bell \approx \sqrt{\nabla^{2}\Phi(\bf)}.$
One can show that augmenting $\bf$ with a multiple of $\bc$ decreases the potential by at least $m^{-o(1)}.$
After $m^{1+o(1)}$ iterations, the algorithm reaches a flow $\bf$ whose potential value is at most $-\O(m).$

The incremental maximum flow algorithm of \cite{BLS23} dynamizes the potential reduction procedure of \cite{CKLPPS22}. To handle an edge insertions, \cite{BLS23} keeps augmenting the current flow with an $m^{o(1)}$-approximate min-ratio cycle until the output cycle cannot make enough, $m^{-o(1)}$, progress. The analysis of this leverages that adding an edge does not affect the feasibility of the current flow and only increases the potential by a constant amount.
Additionally, once a flow of cost at most the given threshold appears in the graph, any $m^{o(1)}$-approximate min-ratio cycle decreases the potential by at least $m^{-o(1)}$ as long as our current flow has its cost larger than the threshold by $\exp(-\O(1))$.
Since the potential value starts at $\otilde(m)$, over the course of the incremental algorithm, there are at most $m^{1+o(1)}$ approximate min-ratio cycle queries. This yields a $n^{1/2+o(1)}$-update time using their adaptive data structure for answering min-ratio cycle queries.

\paragraph{Obstacles.}
In the static case, \cite{CKLPPS22} manages to apply an oblivious data structure for iteratively minimizing the IPM potential due to (a) the existance an optimal flow before initializing the data structure and (b) the stability of the gradient and Hessian of the IPM potential (which are the inputs to the min-ratio cycle data structure).
Consequently, whenever the data structure fails to find a cycle good enough to make progress, it must be the case that some part of the data structure is broken, and we need to fix it actively.
However, (a) does not hold in the incremental case.
That is, whenever the oblivious data structure fails, we cannot distinguish between the case that (1) the data structure fails due to the obliviousness, or (2) the graph does not support the optimal flow because every feasible flow has a large cost.
The issue occurs even when all incremental updates come from an oblivious adversary.

\paragraph{Our Approach.} To obtain our results we depart from prior work in both how we reason about adaptive adversaries in the minimum ratio cycle problem and in what optimization method we use for this dynamic data structure. We elaborate on each of these in the next \cref{sec:overview:adversary} and \cref{sec:overview:iter} and then discuss how they are put together to obtain our results in \cref{sec:overview:altogether}.

\subsection{Adaptive Adversaries and Monotonic Dynamic Min-Ratio Cycle.}
\label{sec:overview:adversary}

We manage to use the oblivious min-ratio cycle data structure by ensuring that edge lengths are \emph{mostly increasing}.
That is, with a different numerical method, we can divide the entire incremental algorithm into $m^{o(1)}$ phases and within each phase, we solve a sequence of slowly changing min-ratio cycle problems with the same gradient $\bg$ and monotonically increasing lengths $\bell.$ 
In \cref{sec:monoMRC} which presents the min-ratio cycle data structure, we work with an adversary where we can assume that all incremental updates are determined before the initialization.
That is, the adversary does not have access to the internal state of our algorithm when deciding which edges to insert next.
Using these facts, we show, in  \cref{sec:monoMRC}, that the updates to the dynamic min-ratio cycle data structure satisfy a weaker form of the \emph{hidden stable-flow chasing} property, which was the critical property that \cite{CKLPPS22} leveraged to show correctness of their min-ratio cycle data structure. While \cite{CKLPPS22} had to periodically rebuild layers of their data structure when the ``lengths" of flows at those layers may have decreased, and this prevented its application to incremental directed maxflow in \cite{BLS23}, our monotonicity property allows us to avoid this issue.
In particular, whenever the data structure fails to output a good cycle, we know for certain that every feasible flow has large congestion.

\subsection{From IPMs to Iterative Refinement and MWU}
\label{sec:overview:iter}

This paper focuses on solving \emph{incremental thresholded $\ell_p$-norm flow} for $p \ge 2.$
This immediately gives an incremental $(1+\eps)$-approximate maximum flow due to the choice of $p.$
Unlike IPM-based maxflow algorithms, $\ell_p$-norm flows can be reduced to approximately solving $m^{o(1)}$ smoothed $\ell_p$-norm flow \emph{residual problems} to $m^{o(1)}$-approximation factors.
To approximately solve each residual problem, we use a multiplicative weight update method (MWU) to reduce the problem to a sequence of $m^{1+o(1)}$ slowly changing $\ell_1$-regression sub-problems, which are equivalent to min-ratio cycles in our case.
The nature of MWU ensures that the $\ell_1$ weights are non-decreasing.
These constraints on how the sequence of min-ratio cycle instances change enable us to use the data structure of \cite{CKLPPS22} to achieve the almost-linear runtime.

The reduction to residual problem is achieved via the iterative refinement framework of \cite{AKPS19}.
Each residual problem asks to find a circulation $\bDelta$
such that
\begin{align}\label{eq:residualOverview}
    \norm{\mR \bDelta}_2 \le m^{o(1)}\text{, }
    \norm{\mW \bDelta}_p \le m^{o(1)}\text{, and }
    \l\bg, \bDelta\r = -1
\end{align}
where $\br, \bw \in \R^E_+$ are edge weights derived from the current solution to \cref{prob:pnormflow}, and we are guaranteed the existence of a circulation whose both norms are at most $1.$

To find a feasible residual solution to Problem \cref{eq:residualOverview}, we use a MWU to reduce the problem to a sequence of $m^{1+o(1)}$ $\ell_1$ sub-problems of finding a circulation $\bc$ such that
\begin{align}\label{eq:l1subpOverview}
    \norm{\mL\bc}_1 \le m^{o(1)} \norm{\bell}_1\text{ and }\l\bg, \bc\r = -1
    \,.
\end{align}
This is equivalent to finding a min-ratio cycle \eqref{eq:MRCOverview} up to scaling the solution so that $\l\bg, \bc\r = -1.$
If we can solve each $\ell_1$ sub-problem to $m^{o(1)}$-approximation, we obtain a feasible residual solution.
Furthermore, MWU ensures that the $\ell_1$ weights are non-decreasing across all $m^{1+o(1)}$ instances.

\subsection{Putting it Altogether}
\label{sec:overview:altogether}
To dynamize the approach mentioned in \cref{sec:overview:iter}, we make the following observation:
\begin{quote}
    If $OPT \le F$, there is a feasible solution $\bDelta^*$ to \eqref{eq:residualOverview} with
    $\|\mR \bDelta\|_2, \|\mW \bDelta\|_p \le 1$
    that satisfies \eqref{eq:l1subpOverview}, i.e., $\|\mL \bDelta^*\|_1 \le \|\bell\|_1$ for any $\bell$ encountered in the algorithm.
\end{quote}
We leverage that $\bDelta^*$ is fixed with respect to $\br, \bw$ and $\bg$.
Thus, whenever the $m^{o(1)}$-approximate min-ratio cycle data structure cannot find a solution to \eqref{eq:l1subpOverview}, we certify that $OPT > F.$
Otherwise, if $OPT \le F$, the MWU method, as well as the iterative refinement procedure, would proceed as desired and output a flow $\bf$ of $\ell_p$-norm energy at most $\cE(\bf) \le F+\eps$ in $p^2 m^{1+o(1)} \log(1/\eps)$-time.

%% file: refinement.tex
\section{Incremental \texorpdfstring{$p$}{p}-Norm Iterative Refinement}
\label{sec:incrIterRefine}

In this section, we discuss the iterative refinement approach to solving $p$-norm flows and how reduce them to incremental MWUs.
The iterative refinement framework of \cite{AKPS19} reduces $p$-norm flows to approximately solving a small number of residual problems.
The original framework requires an estimate on the optimal residual value, which is obtained via binary search.
In our setting, a target objective value is given and we can directly use it to estimate the residual value.
This requires a slight modification to the convergence analysis via measuring the progress towards the target value.

The goal of this section is to prove the following theorem:
\thmIncrPnormflow*

We first define the residual problem, $\cR_{\bf}(\bx)$, that we consider to approximate $\cE(\bf + \bx) - \cE(\bf)$ (\Cref{prob:res}) and provide a known lemma about its approximation quality (\Cref{lem:lpIR}).

\begin{prob}[Residual Problem]\label{prob:res}
For feasible flow $\bf$ in $p$-Norm flow instance $(G, \bg^G, \br^G, \bw^G, \bd)$, we define its \emph{residual problem}, $\cR_{\bf}(\bx)$, (that approximates ) as follows:
\begin{align*}
    \cR_{\bf}(\bx) \defeq \l\bg, \bx\r + \norm{\mR \bx}_2^2 + \norm{\mW \bx}_p^p,
\end{align*}
where $\bg \defeq \bg^G + 2 (\mR^G)^2 \bf + p (\mW^G)^p|\bf|^{p-2}\bf, \br \defeq \sqrt{(\br^G)^2 + 2p^2 (\mW^G)^p |\bf|^{p-2}}, \bw \defeq p \bw^G.$

We often ignore the subscript $\bf$ when it is clear from the context.
\end{prob}

\begin{lemma}[Iterative Refinement, \cite{AKPS19, APS19, AS20}]\label{lem:lpIR}
For any $\bf$ and $\bx$, we have
\begin{align*}
    \cE(\bf + \bx) - \cE(\bf) &\le \cR_f(\bx), \text{ and} \\
    \cE(\bf + \lambda \bx) - \cE(\bf) &\ge \lambda \cR_f(\bx), \text{ for some $\lambda = O(p)$}
\end{align*}
\end{lemma}

Using \cref{lem:lpIR}, we can relate the optimal residual objective value to the gap between the current feasible flow $\bf$ and the target threshold $F.$

\begin{lemma}[Threshold Certification by Residual Value]
\label{lem:resThreshold}
If feasible flow $\bf$ in a $p$-Norm flow instance $(G, \bg^G, \br^G, \bw^G, \bd)$ satisfies $OPT \le F$, then there is a circulation $\bc^*$ with $\cR(\bc^*) \le (F - \cE(\bf)) / \lambda.$
\end{lemma}
\begin{proof}
Let $\bf^*$ be any feasible flow such that $\cE(\bf^*) \le F.$
\cref{lem:lpIR} yields that
\begin{align*}
F - \cE(\bf) \ge \cE(\bf^*) - \cE(\bf) \ge \lambda \cR\left(\frac{\bf^* - \bf}{\lambda}\right).
\end{align*}
The conclusion follows because $\bf^* - \bf$ is a circulation.
\end{proof}

Let $R \defeq (\cE(\bf) - F) / \lambda > 0$ be the residual threshold.
Our goal now is to find a circulation $\bc$ such that $\cR(\bc) \le -R / K$ for some $K = m^{o(1)}.$
We show that we can compute this by solving the following \emph{incremental residual problem}.

\begin{restatable}[Incremental $K$-Approximate Residual Problem]{prob}{probIncrRes}\label{prob:incrRes}
Consider an incremental graph $G = (V, E)$ with at most $m$ edges, a gradient vector $\bg \in \R^E$, $\ell_2$ and $\ell_p$ edge weights $\br, \bw \in \R^E_+$ and a approximation factor $K > 0$.
The \emph{Incremental $K$-Approximate Residual Problem} asks, after the initialization or each edge insertion, to either
\begin{enumerate}
\item certify that there is no circulation $\bc^*$ with $\l\bg, \bc^*\r = -1$, $\|\mR \bc^*\|_2 \le 1$, and $\|\mW \bc^*\|_p \le 1$, or
\item output a circulation $\bc$ such that $\l\bg, \bc\r = -1$, $\|\mR \bc\|_2 \le K$, and $\|\mW \bc\|_p \le K$.
\end{enumerate}
\end{restatable}

In \cref{sec:incrMWU}, we will present an almost-linear time algorithm for \cref{prob:incrRes}.
\begin{restatable}{lemma}{lemIncrRes}\label{lem:incrRes}
For some $K = m^{o(1)}$, there is a randomized algorithm for \cref{prob:incrRes}, denoted as $\cA^{(\ref{lem:incrRes})}$,
that runs in $m^{1+o(1)}$-time and succeeds with high probability against an adaptive adversary.
\end{restatable}
Here, recall that the adversary is only adaptive against the yes/no output of the algorithm, and cannot see the internal randomness, including the flow that is being stored internally. Obviously, our algorithm can output a flow when we get a ``no" output, by computing a high-accuracy $p$-norm minimizing flow, but this is different from the flow being stored internally.

At first glance, the solution guarantee of \cref{prob:incrRes} has nothing to do with the residual problem (\cref{prob:res}) and $R$.
However, since we only need a solution whose objective is better than $m^{-o(1)}(F - \cE(\bf)) / \lambda$, we can reduce incremental $p$-norm regression to the form of \cref{prob:incrRes} using the following pair of lemmas. The first is a straightforward scaling argument.

\begin{lemma}\label{lem:resRatio}
If there's some $\bc^*$ such that $\cR(\bc^*) \le -R$
for some threshold $R > 0$ then,
\begin{align*}
    \frac{\l\bg, \bc^*\r}{\norm{\mR \bc^*}_2} \le -2 \sqrt{R}\text{ and }
    \frac{\l\bg, \bc^*\r}{\norm{\mW \bc^*}_p} \le -R^{(p-1)/p}
\end{align*}
\end{lemma}
\begin{proof}
From the bound that $\cR(\bc^*) \le -R$, we know
\begin{align*}
    \l\bg, \bc^*\r + \norm{\mR \bc^*}_2^2 \le -R
\end{align*}
Thus we know that
\[ \frac{\l \bg, \bc^*\r}{\norm{\mR\bc^*}_2} \le -\frac{R}{\norm{\mR\bc^*}_2} - \norm{\mR\bc^*}_2 \le -2\sqrt{R}, \] by the AM-GM inequality. Similarly, 
\[ \l\bg, \bc^*\r + \norm{\mW \bc^*}_p^p \le -R, \]
and hence
\[ \frac{\l \bg, \bc^*\r}{\norm{\mW\bc^*}_p} \le -\frac{R}{\norm{\mW\bc^*}_p} - \norm{\mW\bc^*}_p^{p-1} \le -R^{\frac{p-1}{p}}, \] where we used the (weighted) AM-GM inequality again.
\end{proof}

The second, given below, is used to argue that the solution quality of \cref{prob:incrRes} translates to a bound on the residual objective value.
\begin{lemma}
\label{lem:resValue}
Given a residual threshold $R > 0$ and the current residual problem $\cR(\bc)$ with gradient $\bg$ and $\ell_2$ and $\ell_p$ edge weights $\br$ and $\bw$, consider a circulation $\bc$ such that $\l\bg, \bc\r = -1$, $\|2\sqrt{R}\mR \bc\|_2 \le K$, and $\|R^{(p-1)/p} \mW \bc\|_p \le K$ for some $K \ge 1$.
We have $\cR((R / (2 K^2)) \bc) \le -R / (6 K^2).$
\end{lemma}
\begin{proof}
From the assumptions of the lemma, we know that
\begin{align*}
    \l\bg, \bc\r = -1\text{, }
    \norm{\mR \bc}_2 \le \frac{K}{2\sqrt{R}}\text{, and }
    \norm{\mW \bc}_p \le K R^{-(p-1)/p}
\end{align*}
Plugging these into the definition of the residual problem yields
\begin{align*}
\cR\left(\frac{R}{2K^2}\bc\right) 
&= \frac{R}{2K^2} \l\bg, \bc\r + \frac{R^2}{4K^4}\norm{\mR \bc}_2^2 + \frac{R^p}{2^pK^{2p}} \norm{\mW \bc}_p^p \\
&\le -\frac{R}{2K^2} + \frac{R^2}{4K^4}\frac{K^2}{4R} + \frac{R^p}{2^pK^{2p}} K^{p} R^{-(p-1)} \\
&= -\frac{R}{2K^2} + \frac{R}{16K^2} + \frac{R}{2^pK^p} \le \frac{-R}{5K^2}\,.
\end{align*}
\end{proof}

Now, we are ready to state \cref{algo:incrPnormflow} which we use to prove \cref{thm:incrPnormflow}.

\begin{algorithm}[!ht]
  \caption{Incremental Thresholded $p$-Norm Flow \label{algo:incrPnormflow}}
  \SetKwProg{Globals}{global variables}{}{}
  \SetKwProg{Proc}{procedure}{}{}
  \Globals{}{
    $p$: the norm considered in \cref{prob:incrPnormflow}. \\
    $\lambda \gets O(p)$: the scaling factor in \cref{lem:lpIR}. \\
    $K \gets m^{o(1)}$: the approximation factor in \cref{lem:incrRes}.
  }
  \Proc{$\IncrPNorm(G, \bd, \bg^G, \br^G, \bw^G, F, \eps)$}{
    Initialize $\bf^{(0)}$ to be the optimal $p$-norm flow on $G$. \\
    $S \defeq O(p K^2 \log (\frac{\cE(\bf^{(0)}) - F}{\eps}))$ \\
    \For{step $t = 0, 1, \ldots, S$}{
        Construct residual problem $(\bg^{(t)}, \br^{(t)}, \bw^{(t)})$ with respect to $\bf^{(t)}.$ \\
        Define residual threshold $R \defeq (\cE(\bf^{(t)}) - F) / \lambda$. \\
        Initialize and run the incremental approximate residual algorithm \\
        \While{$\cA^{(\ref{lem:incrRes})}$ cannot find a circulation satisfying \Cref{lem:incrRes} item 2}{
            Claim that $OPT > F.$ \\
            Wait for a new edge $e.$ \\
            Insert the new edge $e$ to $\cA^{(\ref{lem:incrRes})}$ with
            \begin{align*}
            \bg^{(t)}_e \defeq \bg^G_e
            \text{, }
            \br^{(t)}_e \defeq 2\sqrt{R}\br^G_e\text{, and } \bw^{(t)}_e \defeq R^{(p-1)/p}\bw^G_e   
            \end{align*} \\
        }
        $\cA^{(\ref{lem:incrRes})}$ finds a circulation $\bc$ such that
        \begin{align*}
        \l\bg^{(t)}
        \text{ , }\bc\r = -1, \|2\sqrt{R}\mR^{(t)} \bc\|_2 \le K
        \text{, and }\|R^{(p-1)/p} \mW^{(t)} \bc\|_p \le K\,.
        \end{align*} \\
        Set $\bf^{(t+1)} \gets \bf^{(t)} + \frac{R}{2K^2} \bc$
    }
    \Return $\bf^{(S)}$
  }
\end{algorithm}

To show correctness, we first prove that each step makes an exponential progress. Leveraging this lemma we then prove \cref{thm:incrPnormflow}.

\begin{lemma}[Iterative Refinement Convergence Rate]
\label{lem:convergence}
At the end of any step $t$ in \cref{algo:incrPnormflow},
\begin{align*}
    \cE(\bf^{(t+1)}) - F \le \left(1 - \frac{1}{6K^2 \lambda}\right) \left(\cE(\bf^{(t)}) - F\right)
\end{align*}
\end{lemma}
\begin{proof}
At the end of step $t$, \cref{lem:resValue} ensures that $\cR(\frac{R}{2K^2} \bc) \le -\frac{R}{6K^2}.$
\cref{lem:lpIR} then yields
\begin{align*}
\cE(\bf^{(t+1)}) - \cE(\bf^{(t)})
\le \cR\left(\frac{R}{2K^2} \bc\right)
\le \frac{-1}{6K^2 \lambda} \left(\cE(\bf^{(t)}) - F\right)
\end{align*}
Thus, we have
\begin{align*}
\cE(\bf^{(t+1)}) - F
&= \cE(\bf^{(t+1)}) - \cE(\bf^{(t)}) + \cE(\bf^{(t)}) - F \\
&\le \left(1 - \frac{1}{6K^2 \lambda}\right) \left(\cE(\bf^{(t)}) - F\right)
\end{align*}
\end{proof}

\begin{proof}[Proof of \cref{thm:incrPnormflow}]
We first show that whenever the algorithm outputs that $OPT > F$, this is indeed the case.
Let $t$ be an arbitrary step in which the algorithm outputs that $OPT > F$, let $\bf^{(t)}$ be the flow maintained at that step, and let $\cR(\bx)$ be the corresponding residual problem instance (\cref{prob:res}).
Correctness of \cref{lem:incrRes} ensures that no circulation $\bc^*$ satisfies
\begin{align*}
    \l\bg^{(t)}, \bc^*\r = -1, \|2\sqrt{R}\mR^{(t)} \bc^*\|_2 \le 1, \|R^{(p-1)/p} \mW^{(t)} \bc^*\|_p \le 1
\end{align*}
This fact and \cref{lem:resRatio} then implies that every circulation $\bc^*$ must have $\cR(\bc^*) > -R = (F - \cE(\bf^{(t)})) / \lambda.$
\cref{lem:resThreshold} implies that $OPT > F.$

Now, we analyze the quality of the output, $\bf^{(S)}.$
Applying \cref{lem:convergence} inductively on steps $t$ yields that
\begin{align*}
    \cE(\bf^{(S)}) - F \le \left(1 - \frac{1}{6K^2 \lambda}\right)^S \left(\cE(\bf^{(0)}) - F\right) \le \eps
\end{align*}
by definitions of $S$ and $\lambda$ (\cref{lem:lpIR}).

Finally, the runtime comes directly from $S = m^{o(1)} p \log(\frac{1}{\eps})$ applications of \cref{lem:incrRes}. (Recall that there is an extra factor of $p$ due to bit complexity.)
\end{proof}

\section{Incremental Multiplicative Weight Updates}
\label{sec:incrMWU}

In this section, we present a MWU-based algorithm that solves \cref{prob:incrRes} and proves \cref{lem:incrRes}.
We use \cref{algo:incrRes} to prove \cref{lem:incrRes}.
At a high level, the algorithm outputs the final circulation as an average of $T = \O(m)$ circulations.
At each iteration $i$, we compute a set of $\ell_1$ edge weights $\bell \in \R^E_+$ and compute an approximate min ratio cycle that minimizes $\l\bg, \bc\r / \|\mL \bc\|_1.$
As is too costly to compute each cycle from scratc, we instead apply a dynamic data structure simplified from \cite{CKLPPS22} to compute such cycle efficiently.

\begin{restatable}[Approximate Monotonic Min Ratio Cycle (MRC)]{prob}{probMonoMRC}\label{prob:monoMRC}
Consider a target ratio $\alpha > 0$ and a dynamic graph $G = (V, E)$ with edge gradients $\bg \in \R^E$, and lengths $\bell \in \R^E_+$ under edge insertions and length increases.
In addition, there is a hidden circulation $\bc^*$ not necessarily supported on $G$, i.e., there is possibly some edge not in $G$ with $\bc^*_e \neq 0.$
The problem of \emph{$\kappa$-approximate monotonic min ratio cycle} asks, after the initialization and each edge update,
\begin{enumerate}
\item If $\bc^*$ is not supported on $G$: output any circulation.
\item If $\bc^*$ is supported on $G$: output a circulation $\bc$ such that
\begin{align*}
    \frac{\l\bg, \bc\r}{\norm{\mL \bc}_1} \le \frac{-\alpha}{\kappa},
\end{align*}
given that $\frac{\l \bg, \bc^*\r}{\norm{\mL \bc^*}_1} \le -\alpha$.
\end{enumerate}
\end{restatable}

We can solve \cref{prob:monoMRC} in almost-linear time by using a data structure from \cite{CKLPPS22}. The guarantee of this data structure is stated below and the lemma is proved in \Cref{sec:monoMRC}.

\begin{restatable}[Approximate Monotonic MRC Data Structure]{lemma}{lemMonoMRC}\label{lem:monoMRC}
For some $\kappa = m^{o(1)}$, there is a randomized algorithm, denoted $\cA^{(\ref{lem:monoMRC})}$, that maintains a collection of $s = m^{o(1)}$ spanning trees $T_1, T_2, \ldots, T_s$ and solves the $\kappa$-approximate monotonic MRC problem (\cref{prob:monoMRC}) in $(m+Q)m^{o(1)}$-time where $Q$ is the total number of updates.
In particular, it always output a cycle $\bDelta$ represented by $m^{o(1)}$ off-tree edges and paths on a spanning tree $T_i$ for some $i \in [s]$.
\end{restatable}

We now state \cref{algo:incrRes} which shows \cref{lem:incrRes}. The algorithm runs a MWU algorithm with $\ell_1$-norm subproblems for $\O(m)$ iterations. This is motivated by the $\ell_1$-IPM of \cite{CKLPPS22} for solving min-cost flow, and contrasts with $\ell_2$-norm MWUs of \cite{CKMST11,AKPS19} for approximate maxflow and $p$-norm regression. These $\ell_1$-norm subproblems are min-ratio cycle problems, which we solve by calling a min-ratio cycle data structure. If the output has sufficiently good ratio, the algorithm proceeds to the next iteration. Otherwise, we conclude that $G$ does not support a valid circulation yet, and ask for the next edge insertion.
\begin{algorithm}[!ht]
  \caption{Incremental $K$-Approximate Residual Algorithm \label{algo:incrRes}}
  \SetKwProg{Globals}{global variables}{}{}
  \SetKwProg{Proc}{procedure}{}{}
  \Globals{}{
    $p$: the norm considered in \cref{prob:incrRes}. \\
    $q \gets \min\{\log_2 m, p\}$: the norm which the algorithm handles. \\
    $\kappa \gets m^{o(1)}$: the approximation factor in \cref{lem:monoMRC}. \\
    $K \gets 100 q \kappa$: the approximation factor of this algorithm. \\
    $\alpha \gets K^{1-q} / (10q)$: the target ratio given in \cref{lem:targetRatio}. \\
  }
  \Proc{$\IncrMWU(G, \bg, \br, \bw)$}{
    Initialize $\bc^{(0)} = \mathbf{0}$, $\ba^{(0)} = K m^{-1/2} \br^{-1}$, and $\bb^{(0)} = K m^{-1/q}\bw^{-1}.$ \\
    Initialize edge lengths $\bell^{(0)} \defeq K^{q-2}\br^2\ba^{(0)} + \bw^q(\bb^{(0)})^{q-1}.$ \\
    Initialize a $\kappa$-approximate MRC algorithm $\cA^{(\ref{lem:monoMRC})}$ on $(G, \bg, \bell)$ and target ratio $\alpha.$ \\
    Set $T \gets 100qm$ \\
    \For{iteration $i = 0, 1, \ldots, T$}{
        Define the edge length $\bell^{(i)} \defeq K^{q-2}\br^2\ba^{(i)} + \bw^q(\bb^{(i)})^{q-1}.$\\
        Maintain an estimation $\wt{\bell} \in [\bell^{(i)}, 2 \bell^{(i)}].$\\
        Feed updates to the estimation $\wt{\bell}$ to $\cA^{(\ref{lem:monoMRC})}.$ \\
        \While{$\cA^{(\ref{lem:monoMRC})}$ outputs a cycle of ratio $> -\alpha / \kappa$ \label{line:outputbad}}{
            Claim that there's no circulation $\bc^*$ s.t. 
            \begin{align*}
                \l\bg, \bc^*\r = -1\text{, }
                \norm{\mR\bc^*}_2 \le 1\text{, and }
                \norm{\mW\bc^*}_p \le 1
            \end{align*} \label{line:claim} \\
            Wait for a new edge $e.$ \\
            Set $\ba_e^{(i)} = K m^{-1/2} \br^{-1}_e$ and $\bb_e^{(i)} = K m^{-1/q}\bw^{-1}_e.$ \\
            Insert the new edge $e$ to $\cA^{(\ref{lem:monoMRC})}$ with gradient $\bg_e$ and length $\wt{\bell}_e = \bell_e^{(i)} = K^{q-2}\br^2_e\ba_e^{(i)} + \bw^q_e(\bb^{(i)}_e)^{q-1}.$
        }
        $\cA^{(\ref{lem:monoMRC})}$ outputs a cycle $\bDelta^{(i)}$ s.t. $\l\bg, \bDelta^{(i)}\r / \|\wt{\mL} \bDelta^{(i)}\|_1 \le -\alpha / \kappa.$ \label{line:foundCycle} \\
        Scale $\bDelta^{(i)}$ such that $\l\bg, \bDelta^{(i)}\r = -1$ and $\|\wt{\mL} \bDelta^{(i)}\|_1 \le \kappa / \alpha \le K^q$. \\
        Update $\bc^{(i+1)} \defeq \bc^{(i)} + \frac{1}{T} \bDelta^{(i)}$, $\ba^{(i+1)} \defeq \ba^{(i)} + \frac{1}{T} |\bDelta^{(i)}|$, and $\bb^{(i+1)} \defeq \bb^{(i)} + \frac{1}{T} |\bDelta^{(i)}|$
    }
    \Return $\bc^{(T)}$
  }
\end{algorithm}

To analyze \cref{algo:incrRes}, we keep track of the following two potentials:
\begin{align}
    \Phi^{(i)} &\defeq \norm{\mR \ba^{(i)}}_2^2 \\
    \Psi^{(i)} &\defeq \norm{\mW \bb^{(i)}}_q^q
\end{align}
where $\ba$ and $\bb$ are defined as
\begin{align}
    \ba^{(i)} &\defeq \frac{K}{\sqrt{m}} \br^{-1} + \frac{1}{T} \sum_{j = 0}^{i-1} |\bDelta^{(j)}| \\
    \bb^{(i)} &\defeq \frac{K}{m^{1/q}} \bw^{-1} + \frac{1}{T} \sum_{j = 0}^{i-1} |\bDelta^{(j)}|
\end{align}

Because of the definition of $\ba^{(i)}$ and $\bb^{(i)}$, we know $\ba^{(i)}, \bb^{(i)} \ge |\bc^{(i)}|$ after any iteration $i$ and thus both $\Phi^{(i)} \le \|\mR \bc^{(i)}\|_2^2$ and $\Psi^{(i)} \le \|\mW \bc^{(i)}\|_q^q$.

\begin{lemma}[Initial Potential]
\label{lem:initPot}
Initially, $\Phi^{(0)} = K^2$ and $\Psi^{(0)} = K^q.$
\end{lemma}

\begin{lemma}[Increase in $\Phi$]
\label{lem:phiIncr}
At any iteration $i$, we have $\Phi^{(i+1)} \le \Phi^{(i)} + 3K^2 / T.$
\end{lemma}
\begin{proof}
For simplicity in the presentation, we ignore superscripts $(i)$ in the proof.
Let $\bDelta$ be the cycle output in \cref{line:foundCycle} such that $\|\wt{\mL} \bDelta\|_1 \le K^q$ and $\Phi'$ be the new potential values after updating $\ba' \defeq \ba + \frac{1}{T} |\bDelta|.$
For any edge $e$, we know
\begin{align*}
    K^{q-2} \br_e^2 \ba_e |\bDelta_e| &\le 
    \bell_e |\bDelta_e| \le
    \norm{\mL \bDelta}_1 \le \|\wt{\mL} \bDelta\|_1 \le K^q
\end{align*}
Rearrangement and the monotonicity of $\ba$ yields
\begin{align*}
    \frac{K^2}{m} \frac{|\bDelta_e|}{\ba_e} \le \br_e^2 \ba^2_e \frac{|\bDelta_e|}{\ba_e} = \br_e^2 \ba_e |\bDelta_e| \le K^2
\end{align*}
and therefore $|\bDelta_e| \le m \ba_e \le T \ba_e / 3.$
Applying the definition of $\Phi'$ yields
\begin{align*}
\Phi'
&= \sum_e \br_e^2 (\ba_e + \frac{1}{T}|\bDelta_e|)^2 \\
&\le \sum_e \br_e^2 (\ba^2_e + \frac{3}{T} \ba_e |\bDelta_e|) = \Phi + \frac{3}{T} \sum_e \br_e^2 \ba_e |\bDelta_e| \\
&\le \Phi + \frac{3}{T} K^{2-q} \norm{\mL \bDelta}_1 \le \Phi + \frac{3K^2}{T}
\end{align*}
where we use $(1+x)^2 \le 1+3x$ for $x \in [0,1]$.
\end{proof}

\begin{lemma}[Increase in $\Psi$]
\label{lem:psiIncr}
At any iteration $i$, we have $\Psi^{(i+1)} \le \Psi^{(i)} + 4qK^q / T.$
\end{lemma}
\begin{proof}
For simplicity in the presentation, we ignore superscripts $(i)$ in the proof.
Let $\bDelta$ be the cycle output in \cref{line:foundCycle} such that $\|\wt{\mL} \bDelta\|_1 \le K^q$ and $\Psi'$ be the new potential values after updating $\bb' \defeq \bb + \frac{1}{T} |\bDelta|.$
For any edge $e$, we know
\begin{align*}
    \bw_e^q \bb^{q-1}_e |\bDelta_e| \le 
    \bell_e |\bDelta_e| \le \norm{\mL \bDelta}_1 \le \|\wt{\mL} \bDelta\|_1 
    \le K^q
\end{align*}
Rearrangement and the monotonicity of $\bb$ yields
\begin{align*}
    \frac{K^q}{m} \frac{|\bDelta_e|}{\bb_e} \le \bw_e^q \bb^q_e \frac{|\bDelta_e|}{\bb_e}
    = \bw_e^q \bb^{q -1}_e |\bDelta_e|
    \le K^q
\end{align*}
and therefore $|\bDelta_e| \le m \bb_e \le \frac{T \bb_e}{100q}.$
Applying the definition of $\Psi'$ yields
\begin{align*}
\Psi' 
&= \sum_e \bw_e^q (\bb_e + \frac{1}{T}|\bDelta_e|)^q \\
&\le \sum_e \bw_e^q (\bb_e^q + \frac{4q}{T} \bb_e^{q-1} |\bDelta_e|) = \Psi + \frac{4q}{T} \sum_e \bw_e^q \bb_e^{q-1} |\bDelta_e| \\
&\le \Psi + \frac{4q}{T} \norm{\mL \bDelta}_1 \le \Psi + \frac{4qK^q}{T}
\end{align*}
where we use $(1+x)^q \le 1 + 4qx$ for $x \in [0, 1/100q].$
\end{proof}

\begin{lemma}[Final Potential]
\label{lem:mwuFinal}
When \Cref{algo:incrRes} returns $\bc^{(T)}$, we have $\Phi^{(T)} \le 4 K^2$ and $\Psi^{(T)} \le 5 q K^q.$
Therefore, we have $\|\mR \bc^{(T)}\|_2 \le 2 K$ and $\|\mW \bc^{(T)}\|_p \le \|\mW\bc^{(T)}\|_q \le 2 K.$
\end{lemma}
\begin{proof}
This follows directly from \cref{lem:initPot}, \cref{lem:phiIncr}, and \cref{lem:psiIncr}.
\end{proof}

\begin{lemma}[Existence of a good $\ell_1$ solution]
\label{lem:targetRatio}
If there is a circulation $\bc^*$ such that $\l\bg, \bc^*\r = -1, \|\mR \bc^*\|_2 \le 1$, and $\|\mW \bc^*\|_p \le 1$, we have $\|\mL^{(i)} \bc^*\|_1 \le 20 qK^{q-1}$ at any iteration $i.$
\end{lemma}
\begin{proof}
In this proof, we consider an arbitrary iteration $i$ and ignore the superscripts $(i)$ for simplicity.
By definition and the bounds on $\Phi$ and $\Psi$ (\cref{lem:mwuFinal}), we have by the Cauchy-Schwarz inequality and H\"{o}lder inquality that,
\begin{align*}
\norm{\mL \bc^*}_1
&= K^{q-2}\sum_e \br_e^2 \ba_e |\bc^*_e| + \sum_e \bw_e^q \bb_e^{q-1} |\bc^*_e| \\
&\le K^{q-2} \norm{\mR \ba}_2 \norm{\mR \bc^*}_2 + \norm{\mW^{q-1}\bb^{q-1}}_{q/(q-1)} \norm{\mW \bc^*}_q \\
&= K^{q-2} \norm{\mR \ba}_2 \norm{\mR \bc^*}_2 + \norm{\mW\bb}_{q}^{q-1} \norm{\mW \bc^*}_q \\
&\le K^{q-2} (2K) + (5qK^q)^{(q-1)/q} \cdot 2 \\
&\le 2 K^{q-1} + 10 qK^{q-1} \le 20 qK^{q-1}
\end{align*}
where the second inequality comes from $\|\mW \bc^*\|_q \le m^{1/q - 1/p} \|\mW \bc^*\|_p \le m^{1/\log_2 m} = 2.$
\end{proof}

\begin{proof}[Proof of \Cref{lem:incrRes}]
We first show that if line \ref{line:outputbad} of \Cref{algo:incrPnormflow} occurs, then the claim about $\bc^*$ in the following line \ref{line:claim} is true.
The correctness of \cref{lem:monoMRC} says that there's no circulation $\bc^*$ such that $\l\bg, \bc^*\r = -1$ and $\|\wt{\mL} \bc^*\|_1 \le 1 / \alpha = 20q K^{q-1}.$
Correctness then follows from \cref{lem:targetRatio}.

Next, the quality of the output $\bc^{(T)}$ is established by \cref{lem:mwuFinal}.

Finally, the runtime follows from \cref{lem:monoMRC} and the standard technique of maintaining $\ba, \bb$ and $\wt{\bell}$ using dynamic tree data structures such as link-cut trees. In particular, $\bell$ is monotone and bounded by $m^{O(p)}$ on each edge, so we can maintain a $\wt{\bell} \approx_2 \bell$ with $\O(p)$ changes per change.
\end{proof}

%% file: minratio.tex
\section{Dynamic Min-Ratio Cycle Data Structures}
\label{sec:monoMRC}

In this section, we prove \cref{lem:monoMRC}.
We use the oblivious min-ratio cycle data structure from \cite{CKLPPS22}.
In particular, the data structure can handle adversaries stronger than oblivious ones as long as the update sequence satisfies what is called the \emph{hidden stable-flow chasing} property.
In the problem of approximate monotonic min-ratio cycles, we show that the update sequences satisfy this property and the monotonicity in edge lengths allows us to use a special case of the data structure in \cite{CKLPPS22}.

\begin{restatable}[Hidden Stable-Flow Chasing Updates, Definition 6.1~\cite{CKLPPS22}]{definition}{defHiddenStableFlowChasing}
  \label{def:hiddenStableFlowChasing}
  Consider a dynamic graph $G^{(t)}$ undergoing batches of updates $U^{(1)}, \dots, U^{(t)}, \dots$ consisting of edge insertions and deletions.
  We say the sequences $\bg^{(t)}, \bell^{(t)},$ and $U^{(t)}$ satisfy the \emph{hidden stable-flow chasing} property if there are hidden dynamic circulations $\bc^{(t)}$ and hidden dynamic upper bounds $\bw^{(t)}$ such that the following holds at all stages $t$:
  \begin{enumerate}
  \item $\bc^{(t)}$ is a circulation: $\mB_{G^{(t)}}^\top \bc^{(t)} = 0.$ \label{item:circulation}
  \item
  $\bw^{(t)}$ upper bounds the length of $\bc^{(t)}$: $|\bell^{(t)}_e\bc^{(t)}_e| \le \bw^{(t)}_e$ for all $e \in E(G^{(t)})$.
  \label{item:width}
  \item
  For any edge $e$ in the current graph $G^{(t)}$, and any stage $t' \leq t$, if the edge $e$ was already present in $G^{(t')}$, i.e. $e \in G^{(t)} \setminus \bigcup_{s=t' + 1}^{t} U^{(s)}$, then $\bw^{(t)}_e \le 2\bw^{(t')}_e$.
  
  \label{item:widthstable}
  \end{enumerate}
\end{restatable}

\begin{theorem}[Simpler version of {\cite[Theorem 7.1]{CKLPPS22}}]
  \label{thm:norebuild}
  Let $G = (V, E)$ be a dynamic graph undergoing $\tau$ batches of updates $U^{(1)}, \dots, U^{(\tau)}$ containing \emph{only} edge insertions and deletions with edge gradient $\bg^{(t)}$ and length $\bell^{(t)}$ such that the update sequence satisfies the hidden stable-flow chasing property (\cref{def:hiddenStableFlowChasing}) with hidden dynamic circulation $\bc^{(t)}$ and width $\bw^{(t)}.$
  In addition, $\|\bw^{(t)}\|_1$ is non-decreasing in $t$, i.e., $\|\bw^{(t)}\|_1 \le \|\bw^{(t+1)}\|_1$ for all steps $t$.
  
  There is an algorithm on $G$ that for $d = (\log m)^{1/8}$ maintains a collection of $s = O(\log m)^d$ spanning trees $T_1, T_2, \dots, T_s$ and after each update outputs a circulation $\bDelta$ represented by \\ $\exp(O(\log^{7/8}m\log\log m))$ off-tree edges and paths on some $T_i, i \in [s]$.
  The output circulation $\bDelta$ satisfies $\mB^\top \bDelta = 0$ and for some $\kappa = \exp(-O(\log^{7/8}m\log\log m))$
  \begin{align*}
      \frac{\l\bg^{(t)}, \bDelta\r}{\norm{\mL^{(t)} \bDelta}_1} \le \kappa \frac{\l\bg^{(t)}, \bc^{(t)}\r}{(d+1) \|\bw^{(t)}\|_1}
  \end{align*}
  
  The algorithm succeeds whp.\ with total runtime $(m + Q)m^{o(1)}$ for $Q \defeq \sum_{t=1}^{\tau} \Abs{U^{(t)}} \le \poly(n)$.
\end{theorem}
This follows from \cite[Theorem 7.1]{CKLPPS22} because the terms $\|\bw^{(\mathsf{prev}^{(t)}_i)}\|_1 \le \|\bw^{(t)}\|_1$ by the increasing property, and $\mathsf{prev}^{(t)}_i \le t$.

\begin{proof}[Proof of \cref{lem:monoMRC}]
We apply the dynamic algorithm of \cref{thm:norebuild} to prove the lemma.
In order to do so, we need to show that there is a hidden circulation $\bc^{(t)}$ and width $\bw^{(t)}$ that satisfies the hidden stable-flow chasing property (\cref{def:hiddenStableFlowChasing}).
At any stage $t$, if the monotonic MRC problem instance (\cref{prob:monoMRC}) inserts a new edge $e$ with gradient $\bg_e$ and length $\bell_e$, we create an update batch $U^{(t)}$ consists of a single edge insertion $(e, \bg_e, \bell_e).$
Otherwise, if the problem instance increases the length of an edge from $\bell_e$ to $\bell_e'$, we create an update batch $U^{(t)}$ consists of an edge deletion $e$ followed by an edge insertion $(e, \bg_e, \bell_e').$

At any stage $t$ (including stage $0$, the moment after the initialization), we define the hidden circulation $\bc^{(t)}$ to be $\bc^*$ if every edge in the circulation $\bc^*$ appears in $G^{(t)}.$
Otherwise, we set $\bc^{(t)}$ to be the all zero circulation.
We also define the hidden width $\bw^{(t)}$ on each edge $e$ as
\begin{align*}
    \bw^{(t)}_e \defeq \begin{cases}
        \Abs{\bell^{(t)}_e \bc^*_e}&\text{, if $e \in \supp(\bc^*)$} \\
        0&\text{, otherwise.}
    \end{cases}
\end{align*}

Now, we verify each condition of \cref{def:hiddenStableFlowChasing} one at a time.
\cref{item:circulation} holds at any stage $t$ because both $\bc^*$ and $\mathbf{0}$ are circulations.
For any edge $e \in \supp(\bc^*)$, $\bc^{(t)}_e$ is either $0$ or $\bc^*_e$ and $\bw^{(t)}_e \ge |\bell^{(t)}\bc^{(t)}_e|$ holds.
For any other edge $e \not\in \supp(\bc^*)$, both $\bc^{(t)}_e$ and $\bw^{(t)}_e$ are always $0$ and \cref{item:width} follows.
\cref{item:widthstable} follows from the fact that $\bc^*$ is fixed beforehand and that an edge is updated when its length doubles.
Also, $\|\bw^{(t)}\|_1$ is non-decreasing in $t$ because edge lengths are non-decreasing.

Finally, we need to show that the output after each update is always valid.
At stage $t$, if $\supp(\bc^*)\not\subseteq G^{(t)}$, \cref{prob:monoMRC} does not care what we output.
If $\supp(\bc^*) \subseteq G^{(t)}$, we know $\bg^{(t)} = \bg$, $\bc^{(t)} = \bc^*$ and $\|\bw^{(t)}\|_1 = \|\mL^{(t)} \bc^*\|_1$, and the output circulation $\bDelta$ satisfies
\begin{align*}
\frac{\l\bg, \bDelta\r}{\norm{\mL^{(t)} \bDelta}_1}
\le \kappa \frac{\l\bg, \bc^{(t)}\r}{(d+1) \|\bw^{(t)}\|_1}
= \frac{\kappa}{d+1} \frac{\l\bg, \bc^*\r}{\|\mL^{(t)} \bc^*\|_1}
\end{align*}
from \cref{thm:norebuild}.
The conclusion follows because $d = (\log m)^{1/8}$.
\end{proof}

%% file: refs.bib
@inproceedings{GoranciHP17a,
  author       = {Gramoz Goranci and
                  Monika Henzinger and
                  Pan Peng},
  title        = {The Power of Vertex Sparsifiers in Dynamic Graph Algorithms},
  booktitle    = {{ESA}},
  series       = {LIPIcs},
  volume       = {87},
  pages        = {45:1--45:14},
  publisher    = {Schloss Dagstuhl - Leibniz-Zentrum f{\"{u}}r Informatik},
  year         = {2017}
}

@inproceedings{das2022near,
  title={A near-optimal offline algorithm for dynamic all-pairs shortest paths in planar digraphs},
  author={Das, Debarati and Gutenberg, Maximilian Probst and Wulff-Nilsen, Christian},
  booktitle={Proceedings of the 2022 Annual ACM-SIAM Symposium on Discrete Algorithms (SODA)},
  pages={3482--3495},
  year={2022},
  organization={SIAM}
}

@inproceedings{DurfeeGGP19,
  author       = {David Durfee and
                  Yu Gao and
                  Gramoz Goranci and
                  Richard Peng},
  title        = {Fully dynamic spectral vertex sparsifiers and applications},
  booktitle    = {{STOC}},
  pages        = {914--925},
  publisher    = {{ACM}},
  year         = {2019}
}

@article{LiPYZ20,
  author       = {Huan Li and
                  Stacy Patterson and
                  Yuhao Yi and
                  Zhongzhi Zhang},
  title        = {Maximizing the Number of Spanning Trees in a Connected Graph},
  journal      = {{IEEE} Trans. Inf. Theory},
  volume       = {66},
  number       = {2},
  pages        = {1248--1260},
  year         = {2020}
}

@inproceedings{DinitzV94,
  author       = {Yefim Dinitz and
                  Alek Vainshtein},
  title        = {The connectivity carcass of a vertex subset in a graph and its incremental
                  maintenance},
  booktitle    = {{STOC}},
  pages        = {716--725},
  publisher    = {{ACM}},
  year         = {1994}
}

@article{DinitzW98,
  author       = {Yefim Dinitz and
                  Jeffery R. Westbrook},
  title        = {Maintaining the Classes of 4-Edge-Connectivity in a Graph On-Line},
  journal      = {Algorithmica},
  volume       = {20},
  number       = {3},
  pages        = {242--276},
  year         = {1998}
}

@inproceedings{DinitzV95,
  author       = {Yefim Dinitz and
                  Alek Vainshtein},
  title        = {Locally Orientable Graphs, Cell Structures, and a New Algorithm for
                  the Incremental Maintenance of Connectivity Carcasses},
  booktitle    = {{SODA}},
  pages        = {302--311},
  publisher    = {{ACM/SIAM}},
  year         = {1995}
}

@inproceedings{ChalermsookDL+20,
  title={Vertex sparsification for edge connectivity},
  author={Chalermsook, Parinya and Das, Syamantak and Kook, Yunbum and Laekhanukit, Bundit and Liu, Yang P and Peng, Richard and Sellke, Mark and Vaz, Daniel},
  booktitle={Proceedings of the 2021 ACM-SIAM Symposium on Discrete Algorithms (SODA)},
  pages={1206--1225},
  year={2021},
  organization={SIAM}
}

@article{WestbrookT92,
  author       = {Jeffery R. Westbrook and
                  Robert Endre Tarjan},
  title        = {Maintaining Bridge-Connected and Biconnected Components On-Line},
  journal      = {Algorithmica},
  volume       = {7},
  number       = {5{\&}6},
  pages        = {433--464},
  year         = {1992}
}

@inproceedings{GalilIb91,
  author       = {Zvi Galil and
                  Giuseppe F. Italiano},
  title        = {Fully Dynamic Algorithms for Edge-Connectivity Problems (Extended
                  Abstract)},
  booktitle    = {{STOC}},
  pages        = {317--327},
  publisher    = {{ACM}},
  year         = {1991}
}

@inproceedings{Frederickson91,
  author       = {Greg N. Frederickson},
  title        = {Ambivalent Data Structures for Dynamic 2-Edge-Connectivity and k Smallest
                  Spanning Trees},
  booktitle    = {{FOCS}},
  pages        = {632--641},
  publisher    = {{IEEE} Computer Society},
  year         = {1991}
}

@article{EppsteinGIN97,
  author       = {David Eppstein and
                  Zvi Galil and
                  Giuseppe F. Italiano and
                  Amnon Nissenzweig},
  title        = {Sparsification - a technique for speeding up dynamic graph algorithms},
  journal      = {J. {ACM}},
  volume       = {44},
  number       = {5},
  pages        = {669--696},
  year         = {1997}
}

@article{HenzingerK97,
  title={Fully dynamic 2-edge connectivity algorithm in polylogarithmic time per operation},
  author={Henzinger, Monika Rauch and King, Valerie},
  journal={SRC Technical Note},
  volume={4},
  year={1997}
}

@inproceedings{Thorup00,
  author       = {Mikkel Thorup},
  title        = {Near-optimal fully-dynamic graph connectivity},
  booktitle    = {{STOC}},
  pages        = {343--350},
  publisher    = {{ACM}},
  year         = {2000}
}

@article{HolmLT01,
  author       = {Jacob Holm and
                  Kristian de Lichtenberg and
                  Mikkel Thorup},
  title        = {Poly-logarithmic deterministic fully-dynamic algorithms for connectivity,
                  minimum spanning tree, 2-edge, and biconnectivity},
  journal      = {J. {ACM}},
  volume       = {48},
  number       = {4},
  pages        = {723--760},
  year         = {2001}
}

@inproceedings{HolmRT18,
  author       = {Jacob Holm and
                  Eva Rotenberg and
                  Mikkel Thorup},
  title        = {Dynamic Bridge-Finding in \emph{{\~{O}}}(log\({}^{\mbox{2}}\) \emph{n})
                  Amortized Time},
  booktitle    = {{SODA}},
  pages        = {35--52},
  publisher    = {{SIAM}},
  year         = {2018}
}

@inproceedings{GalilI91,
  author       = {Zvi Galil and
                  Giuseppe F. Italiano},
  title        = {Fully Dynamic Algorithms for Edge-Connectivity Problems (Extended
                  Abstract)},
  booktitle    = {{STOC}},
  pages        = {317--327},
  publisher    = {{ACM}},
  year         = {1991}
}

@inproceedings{JinS21,
  author       = {Wenyu Jin and
                  Xiaorui Sun},
  title        = {Fully Dynamic s-t Edge Connectivity in Subpolynomial Time (Extended
                  Abstract)},
  booktitle    = {{FOCS}},
  pages        = {861--872},
  publisher    = {{IEEE}},
  year         = {2021}
}

@inproceedings{BKS23mwu,
  title={Dynamic algorithms for packing-covering lPS via multiplicative weight updates},
  author={Bhattacharya, Sayan and Kiss, Peter and Saranurak, Thatchaphol},
  booktitle={Proceedings of the 2023 Annual ACM-SIAM Symposium on Discrete Algorithms (SODA)},
  pages={1--47},
  year={2023},
  organization={SIAM}
}

@article{BlikstadK23,
  author       = {Joakim Blikstad and
                  Peter Kiss},
  title        = {Incremental (1-{\(\epsilon\)})-approximate dynamic matching in O(poly(1/{\(\epsilon\)}))
                  update time},
  journal      = {CoRR},
  volume       = {abs/2302.08432},
  year         = {2023}
}

@article{BKS23,
  title={Dynamic $(1+\epsilon)$-Approximate Matching Size in Truly Sublinear Update Time},
  author={Bhattacharya, Sayan and Kiss, Peter and Saranurak, Thatchaphol},
  journal={arXiv preprint arXiv:2302.05030},
  year={2023}
}

@article{thorup2007fully,
  title={Fully-dynamic min-cut},
  author={Thorup, Mikkel},
  journal={Combinatorica},
  volume={27},
  number={1},
  pages={91--127},
  year={2007},
  publisher={Springer}
}

@inproceedings{LeMSW22,
	author    = {Hung Le and
	Lazar Milenkovic and
	Shay Solomon and
	Virginia Vassilevska Williams},
	title     = {Dynamic Matching Algorithms Under Vertex Updates},
	booktitle = {{ITCS}},
	series    = {LIPIcs},
	volume    = {215},
	pages     = {96:1--96:24},
	publisher = {Schloss Dagstuhl - Leibniz-Zentrum f{\"{u}}r Informatik},
	year      = {2022}
}

@inproceedings{BosekLSZ14,
	author    = {Bartlomiej Bosek and
	Dariusz Leniowski and
	Piotr Sankowski and
	Anna Zych},
	title     = {Online Bipartite Matching in Offline Time},
	booktitle = {{FOCS}},
	pages     = {384--393},
	publisher = {{IEEE} Computer Society},
	year      = {2014}
}

@article{BernsteinHR19,
	author    = {Aaron Bernstein and
	Jacob Holm and
	Eva Rotenberg},
	title     = {Online Bipartite Matching with Amortized \emph{O}(log \({}^{\mbox{2}}\)
	\emph{n}) Replacements},
	journal   = {J. {ACM}},
	volume    = {66},
	number    = {5},
	pages     = {37:1--37:23},
	year      = {2019}
}

@article{BaswanaGS15,
	author    = {Surender Baswana and
	Manoj Gupta and
	Sandeep Sen},
	title     = {Fully Dynamic Maximal Matching in O(log n) Update Time},
	journal   = {{SIAM} J. Comput.},
	volume    = {44},
	number    = {1},
	pages     = {88--113},
	year      = {2015}
}

@inproceedings{Solomon16,
	author    = {Shay Solomon},
	title     = {Fully Dynamic Maximal Matching in Constant Update Time},
	booktitle = {{FOCS}},
	pages     = {325--334},
	publisher = {{IEEE} Computer Society},
	year      = {2016}
}

@inproceedings{BhattacharyaHN16,
	author    = {Sayan Bhattacharya and
	Monika Henzinger and
	Danupon Nanongkai},
	title     = {New deterministic approximation algorithms for fully dynamic matching},
	booktitle = {{STOC}},
	pages     = {398--411},
	publisher = {{ACM}},
	year      = {2016}
}

@inproceedings{ArarCCSW18,
	author    = {Moab Arar and
	Shiri Chechik and
	Sarel Cohen and
	Cliff Stein and
	David Wajc},
	title     = {Dynamic Matching: Reducing Integral Algorithms to Approximately-Maximal
	Fractional Algorithms},
	booktitle = {{ICALP}},
	series    = {LIPIcs},
	volume    = {107},
	pages     = {7:1--7:16},
	publisher = {Schloss Dagstuhl - Leibniz-Zentrum f{\"{u}}r Informatik},
	year      = {2018}
}

@inproceedings{CharikarS18,
	author    = {Moses Charikar and
	Shay Solomon},
	title     = {Fully Dynamic Almost-Maximal Matching: Breaking the Polynomial Worst-Case
	Time Barrier},
	booktitle = {{ICALP}},
	series    = {LIPIcs},
	volume    = {107},
	pages     = {33:1--33:14},
	publisher = {Schloss Dagstuhl - Leibniz-Zentrum f{\"{u}}r Informatik},
	year      = {2018}
}

@article{BernsteinFH21,
	author    = {Aaron Bernstein and
	Sebastian Forster and
	Monika Henzinger},
	title     = {A Deamortization Approach for Dynamic Spanner and Dynamic Maximal
	Matching},
	journal   = {{ACM} Trans. Algorithms},
	volume    = {17},
	number    = {4},
	pages     = {29:1--29:51},
	year      = {2021}
}

@inproceedings{BhattacharyaK19,
	author    = {Sayan Bhattacharya and
	Janardhan Kulkarni},
	title     = {Deterministically Maintaining a $(2+\epsilon)$-Approximate
	Minimum Vertex Cover in $O(1/\epsilon^2)$ Amortized Update Time},
	booktitle = {{SODA}},
	pages     = {1872--1885},
	publisher = {{SIAM}},
	year      = {2019}
}

@inproceedings{BehnezhadDHSS19,
	author    = {Soheil Behnezhad and
	Mahsa Derakhshan and
	MohammadTaghi Hajiaghayi and
	Cliff Stein and
	Madhu Sudan},
	title     = {Fully Dynamic Maximal Independent Set with Polylogarithmic Update
	Time},
	booktitle = {{FOCS}},
	pages     = {382--405},
	publisher = {{IEEE} Computer Society},
	year      = {2019}
}

@inproceedings{ChechikZ19,
	author    = {Shiri Chechik and
	Tianyi Zhang},
	title     = {Fully Dynamic Maximal Independent Set in Expected Poly-Log Update
	Time},
	booktitle = {{FOCS}},
	pages     = {370--381},
	publisher = {{IEEE} Computer Society},
	year      = {2019}
}

@inproceedings{Wajc20,
	author    = {David Wajc},
	title     = {Rounding dynamic matchings against an adaptive adversary},
	booktitle = {{STOC}},
	pages     = {194--207},
	publisher = {{ACM}},
	year      = {2020}
}

@inproceedings{BhattacharyaK21,
	author    = {Sayan Bhattacharya and
	Peter Kiss},
	title     = {Deterministic Rounding of Dynamic Fractional Matchings},
	booktitle = {{ICALP}},
	series    = {LIPIcs},
	volume    = {198},
	pages     = {27:1--27:14},
	publisher = {Schloss Dagstuhl - Leibniz-Zentrum f{\"{u}}r Informatik},
	year      = {2021}
}

@inproceedings{Kiss21,
  title={Improving update times of dynamic matching algorithms from amortized to worst case},
  author={Kiss, Peter},
  booktitle={Proceedings of the 13th Innovations in Theoretical Computer Science Conference (ITCS)},
  pages={94},
  year={2022}
}

@inproceedings{GuptaP13,
	author    = {Manoj Gupta and
	Richard Peng},
	title     = {Fully Dynamic {(1+} e)-Approximate Matchings},
	booktitle = {{FOCS}},
	pages     = {548--557},
	publisher = {{IEEE} Computer Society},
	year      = {2013}
}

@inproceedings{PelegS16,
	author    = {David Peleg and
	Shay Solomon},
	title     = {Dynamic $(1 + \epsilon)$-Approximate Matchings: {A} Density-Sensitive
	Approach},
	booktitle = {{SODA}},
	pages     = {712--729},
	publisher = {{SIAM}},
	year      = {2016}
}

@inproceedings{BernsteinS15,
	author    = {Aaron Bernstein and
	Cliff Stein},
	title     = {Fully Dynamic Matching in Bipartite Graphs},
	booktitle = {{ICALP} {(1)}},
	series    = {Lecture Notes in Computer Science},
	volume    = {9134},
	pages     = {167--179},
	publisher = {Springer},
	year      = {2015}
}

@inproceedings{BernsteinS16,
	author    = {Aaron Bernstein and
	Cliff Stein},
	title     = {Faster Fully Dynamic Matchings with Small Approximation Ratios},
	booktitle = {{SODA}},
	pages     = {692--711},
	publisher = {{SIAM}},
	year      = {2016}
}

@inproceedings{GrandoniSSU22,
	author    = {Fabrizio Grandoni and
	Chris Schwiegelshohn and
	Shay Solomon and
	Amitai Uzrad},
	title     = {Maintaining an {EDCS} in General Graphs: Simpler, Density-Sensitive
	and with Worst-Case Time Bounds},
	booktitle = {{SOSA}},
	pages     = {12--23},
	publisher = {{SIAM}},
	year      = {2022}
}

@inproceedings{BehnezhadLM20,
	author    = {Soheil Behnezhad and
	Jakub Lacki and
	Vahab S. Mirrokni},
	title     = {Fully Dynamic Matching: Beating 2-Approximation in $\Delta^\epsilon$
Update Time},
	booktitle = {{SODA}},
	pages     = {2492--2508},
	publisher = {{SIAM}},
	year      = {2020}
}

@inproceedings{BehnezhadK22,
	author    = {Soheil Behnezhad and
	Sanjeev Khanna},
	title     = {New Trade-Offs for Fully Dynamic Matching via Hierarchical {EDCS}},
	booktitle = {{SODA}},
	pages     = {3529--3566},
	publisher = {{SIAM}},
	year      = {2022}
}

@inproceedings{RoghaniSW22,
	author    = {Mohammad Roghani and
	Amin Saberi and
	David Wajc},
	title     = {Beating the Folklore Algorithm for Dynamic Matching},
	booktitle = {{ITCS}},
	series    = {LIPIcs},
	volume    = {215},
	pages     = {111:1--111:23},
	publisher = {Schloss Dagstuhl - Leibniz-Zentrum f{\"{u}}r Informatik},
	year      = {2022}
}

@inproceedings{ItalianoNSW11,
  author       = {Giuseppe F. Italiano and
                  Yahav Nussbaum and
                  Piotr Sankowski and
                  Christian Wulff{-}Nilsen},
  title        = {Improved algorithms for min cut and max flow in undirected planar
                  graphs},
  booktitle    = {{STOC}},
  pages        = {313--322},
  publisher    = {{ACM}},
  year         = {2011}
}

@article{AssadiBKL22,
	author    = {Sepehr Assadi and
	Soheil Behnezhad and
	Sanjeev Khanna and
	Huan Li},
	title     = {On Regularity Lemma and Barriers in Streaming and Dynamic Matching},
	journal   = {CoRR},
	volume    = {abs/2207.09354},
	year      = {2022}
}

@inproceedings{Sankowski07,
	author    = {Piotr Sankowski},
	title     = {Faster dynamic matchings and vertex connectivity},
	booktitle = {{SODA}},
	pages     = {118--126},
	publisher = {{SIAM}},
	year      = {2007}
}

@inproceedings{BrandNS19,
	author    = {Brand, Jan van den and
	Danupon Nanongkai and
	Thatchaphol Saranurak},
	title     = {Dynamic Matrix Inverse: Improved Algorithms and Matching Conditional
	Lower Bounds},
	booktitle = {{FOCS}},
	pages     = {456--480},
	publisher = {{IEEE} Computer Society},
	year      = {2019}
}

@inproceedings{HenzingerKNS15,
	author    = {Monika Henzinger and
	Sebastian Krinninger and
	Danupon Nanongkai and
	Thatchaphol Saranurak},
	title     = {Unifying and Strengthening Hardness for Dynamic Problems via the Online
	Matrix-Vector Multiplication Conjecture},
	booktitle = {{STOC}},
	pages     = {21--30},
	publisher = {{ACM}},
	year      = {2015}
}

@inproceedings{Dahlgaard16,
	author    = {Soren Dahlgaard},
	title     = {On the Hardness of Partially Dynamic Graph Problems and Connections
	to Diameter},
	booktitle = {{ICALP}},
	series    = {LIPIcs},
	volume    = {55},
	pages     = {48:1--48:14},
	publisher = {Schloss Dagstuhl - Leibniz-Zentrum f{\"{u}}r Informatik},
	year      = {2016}
}

@article{KumarG03,
	author    = {S. Kumar and
	P. Gupta},
	title     = {An Incremental Algorithm for the Maximum Flow Problem},
	journal   = {J. Math. Model. Algorithms},
	volume    = {2},
	number    = {1},
	pages     = {1--16},
	year      = {2003}
}

@inproceedings{LS14,
  author    = {Yin Tat Lee and
               Aaron Sidford},
  title     = {Path Finding Methods for Linear Programming: Solving Linear Programs
               in $\widetilde{O}(\sqrt{\mathrm{rank}})$ Iterations and Faster Algorithms for Maximum Flow},
  booktitle = {{FOCS}},
  pages     = {424--433},
  publisher = {{IEEE} Computer Society},
  year      = {2014}
}

@inproceedings{BLNPSSSW20,
  title={Bipartite matching in nearly-linear time on moderately dense graphs},
  author={Brand, Jan van den and Lee, Yin-Tat and Nanongkai, Danupon and Peng, Richard and Saranurak, Thatchaphol and Sidford, Aaron and Song, Zhao and Wang, Di},
  booktitle={2020 IEEE 61st Annual Symposium on Foundations of Computer Science (FOCS)},
  pages={919--930},
  year={2020},
  organization={IEEE}
}

@inproceedings{GLP21,
  author    = {Yu Gao and
               Yang P. Liu and
               Richard Peng},
  title     = {Fully Dynamic Electrical Flows: Sparse Maxflow Faster Than Goldberg-Rao},
  booktitle = {{FOCS}},
  pages     = {516--527},
  publisher = {{IEEE}},
  year      = {2021}
}

@inproceedings{BGJLLPS22,
  author    = {Brand, Jan van den and
               Yu Gao and
               Arun Jambulapati and
               Yin Tat Lee and
               Yang P. Liu and
               Richard Peng and
               Aaron Sidford},
  title     = {Faster maxflow via improved dynamic spectral vertex sparsifiers},
  booktitle = {{STOC}},
  pages     = {543--556},
  publisher = {{ACM}},
  year      = {2022}
}

@inproceedings{AMV21,
  author    = {Kyriakos Axiotis and
               Aleksander Madry and
               Adrian Vladu},
  title     = {Faster Sparse Minimum Cost Flow by Electrical Flow Localization},
  booktitle = {{FOCS}},
  pages     = {528--539},
  publisher = {{IEEE}},
  year      = {2021}
}

@inproceedings{CMSV17,
  author    = {Michael B. Cohen and
               Aleksander M\k{a}dry and
               Piotr Sankowski and
               Adrian Vladu},
  title     = {Negative-Weight Shortest Paths and Unit Capacity Minimum Cost Flow
               in $\widetilde{O}(m^{10/7})$ Time (Extended
               Abstract)},
  booktitle = {{SODA}},
  pages     = {752--771},
  publisher = {{SIAM}},
  year      = {2017}
}

@inproceedings{GRST21,
  author    = {Gramoz Goranci and
               Harald R{\"{a}}cke and
               Thatchaphol Saranurak and
               Zihan Tan},
  title     = {The Expander Hierarchy and its Applications to Dynamic Graph Algorithms},
  booktitle = {{SODA}},
  pages     = {2212--2228},
  publisher = {{SIAM}},
  year      = {2021}
}

@inproceedings{JJST22,
  author    = {Arun Jambulapati and
               Yujia Jin and
               Aaron Sidford and
               Kevin Tian},
  title     = {Regularized Box-Simplex Games and Dynamic Decremental Bipartite Matching},
  booktitle = {{ICALP}},
  series    = {LIPIcs},
  volume    = {229},
  pages     = {77:1--77:20},
  publisher = {Schloss Dagstuhl - Leibniz-Zentrum f{\"{u}}r Informatik},
  year      = {2022}
}

@inproceedings{Gupta14,
  author    = {Manoj Gupta},
  title     = {Maintaining Approximate Maximum Matching in an Incremental Bipartite
               Graph in Polylogarithmic Update Time},
  booktitle = {{FSTTCS}},
  series    = {LIPIcs},
  volume    = {29},
  pages     = {227--239},
  publisher = {Schloss Dagstuhl - Leibniz-Zentrum f{\"{u}}r Informatik},
  year      = {2014}
}

@inproceedings{AKPS19,
  title={Iterative refinement for $\ell_p$-norm regression},
  author={Adil, Deeksha and Kyng, Rasmus and Peng, Richard and Sachdeva, Sushant},
  booktitle={Proceedings of the Thirtieth Annual ACM-SIAM Symposium on Discrete Algorithms},
  pages={1405--1424},
  year={2019},
  organization={SIAM}
}

@inproceedings{KPSW19,
  title = {Flows in Almost Linear Time via Adaptive Preconditioning},
  booktitle = {Proceedings of the 51st {{Annual ACM SIGACT Symposium}} on {{Theory}} of {{Computing}}},
  author = {Kyng, Rasmus and Peng, Richard and Sachdeva, Sushant and Wang, Di},
  year = {2019},
  pages = {902--913}
}

@inproceedings{bernstein2020deterministic,
  title={Deterministic decremental reachability, SCC, and shortest paths via directed expanders and congestion balancing},
  author={Bernstein, Aaron and Gutenberg, Maximilian Probst and Saranurak, Thatchaphol},
  booktitle={2020 IEEE 61st Annual Symposium on Foundations of Computer Science (FOCS)},
  pages={1123--1134},
  year={2020},
  organization={IEEE}
}

@inproceedings{M16,
  author    = {Aleksander M{\k{a}}dry},
  title     = {Computing Maximum Flow with Augmenting Electrical Flows},
  booktitle = {57th {IEEE} Annual Symposium on Foundations of Computer Science, {FOCS}
               2016, 9-11 October 2016, Hyatt Regency, New Brunswick, New Jersey,
               {USA}},
  pages     = {593--602},
  publisher = {{IEEE} Computer Society},
  year      = {2016},
  note      = {Available at~\url{https://arxiv.org/abs/1608.06016}}
}

@inproceedings{KLS20,
  author    = {Tarun Kathuria and
               Yang P. Liu and
               Aaron Sidford},
  title     = {Unit Capacity Maxflow in Almost {$O(m^{4/3})$}
               Time},
  booktitle = {61st {IEEE} Annual Symposium on Foundations of Computer Science, {FOCS}
               2020, Durham, NC, USA, November 16-19, 2020},
  pages     = {119--130},
  publisher = {{IEEE}},
  year      = {2020},
  url       = {https://doi.org/10.1109/FOCS46700.2020.00020},
  doi       = {10.1109/FOCS46700.2020.00020},
  bibsource = {dblp computer science bibliography, https://dblp.org}
}

@inproceedings{CKMST11,
  author    = {Paul Christiano and
               Jonathan A. Kelner and
               Aleksander M{\k{a}}dry and
               Daniel A. Spielman and
               Shang{-}Hua Teng},
  title     = {Electrical flows, {L}aplacian systems, and faster approximation of maximum flow in undirected graphs},
  booktitle = {Proceedings of the 43rd {ACM} Symposium on Theory of Computing, {STOC}
               2011, San Jose, CA, USA, June 6-8 2011},
  pages     = {273--282},
  publisher = {{ACM}},
  year      = {2011},
  note = {Available at \url{https://arxiv.org/abs/1010.2921}}
}

@inproceedings{ST04,
  author    = {Daniel A. Spielman and
               Shang{-}Hua Teng},
  title     = {Nearly-linear time algorithms for graph partitioning, graph sparsification, and solving linear systems},
  booktitle = {Proceedings of the 36th Annual {ACM} Symposium on Theory of Computing, {STOC} 2004,
               Chicago, IL, USA, June 13-16, 2004},
  pages     = {81--90},
  year      = {2004},
  note      = {Available at \url{https://arxiv.org/abs/0809.3232}, 
  \url{https://arxiv.org/abs/0808.4134},
  \url{https://arxiv.org/abs/cs/0607105}}
}

@inproceedings{DS08,
  title={Faster approximate lossy generalized flow via interior point algorithms},
  author={Daitch, Samuel I and Spielman, Daniel A},
  booktitle={Proceedings of the fortieth annual ACM symposium on Theory of computing},
  pages={451--460},
  year={2008}
}

@inproceedings{BLLSSSW21,
  author    = { Brand, Jan van den and
               Yin Tat Lee and
               Yang P. Liu and
               Thatchaphol Saranurak and
               Aaron Sidford and
               Zhao Song and
               Di Wang},
  title     = {Minimum cost flows, MDPs, and $\ell_1$-regression
               in nearly linear time for dense instances},
  booktitle = {{STOC}},
  pages     = {859--869},
  publisher = {{ACM}},
  year      = {2021}
}

@inproceedings{CGHPS20,
  title={Fast dynamic cuts, distances and effective resistances via vertex sparsifiers},
  author={Chen, Li and Goranci, Gramoz and Henzinger, Monika and Peng, Richard and Saranurak, Thatchaphol},
  booktitle={2020 IEEE 61st Annual Symposium on Foundations of Computer Science (FOCS)},
  pages={1135--1146},
  year={2020},
  organization={IEEE}
}

@inproceedings{KLOS14,
  author    = {Jonathan A. Kelner and
               Yin Tat Lee and
               Lorenzo Orecchia and
               Aaron Sidford},
  editor    = {Chandra Chekuri},
  title     = {An Almost-Linear-Time Algorithm for Approximate Max Flow in Undirected
               Graphs, and its Multicommodity Generalizations},
  booktitle = {Proceedings of the Twenty-Fifth Annual {ACM-SIAM} Symposium on Discrete
               Algorithms, {SODA} 2014, Portland, Oregon, USA, January 5-7, 2014},
  pages     = {217--226},
  publisher = {{SIAM}},
  year      = {2014},
  url       = {https://doi.org/10.1137/1.9781611973402.16},
  doi       = {10.1137/1.9781611973402.16},
  bibsource = {dblp computer science bibliography, https://dblp.org}
}

@inproceedings{S13,
  author    = {Jonah Sherman},
  title     = {Nearly Maximum Flows in Nearly Linear Time},
  booktitle = {54th Annual {IEEE} Symposium on Foundations of Computer Science, {FOCS}
               2013, 26-29 October, 2013, Berkeley, CA, {USA}},
  pages     = {263--269},
  publisher = {{IEEE} Computer Society},
  year      = {2013},
  url       = {https://doi.org/10.1109/FOCS.2013.36},
  doi       = {10.1109/FOCS.2013.36},
  bibsource = {dblp computer science bibliography, https://dblp.org}
}

@inproceedings{M13,
  title={Navigating central path with electrical flows: From flows to matchings, and back},
  author={M{\k{a}}dry, Aleksander},
  booktitle={2013 IEEE 54th Annual Symposium on Foundations of Computer Science},
  pages={253--262},
  year={2013},
  organization={IEEE}
}

@inproceedings{P16,
  title={Approximate undirected maximum flows in $O(m \mathrm{polylog}(n))$ time},
  author={Peng, Richard},
  booktitle={Proceedings of the twenty-seventh annual ACM-SIAM symposium on Discrete algorithms},
  pages={1862--1867},
  year={2016},
  organization={SIAM}
}

@inproceedings{S17,
  title={Area-convexity, $\ell_{\infty}$ regularization, and undirected multicommodity flow},
  author={Sherman, Jonah},
  booktitle={Proceedings of the 49th Annual ACM SIGACT Symposium on Theory of Computing},
  pages={452--460},
  year={2017}
}

@inproceedings{CKLPPS22,
  title={Maximum flow and minimum-cost flow in almost-linear time},
  author={Chen, Li and Kyng, Rasmus and Liu, Yang P and Peng, Richard and Probst Gutenberg, Maximilian and Sachdeva, Sushant},
  booktitle={2022 IEEE 63rd Annual Symposium on Foundations of Computer Science (FOCS)},
  pages={612--623},
  year={2022},
  organization={IEEE},
  note={\url{https://arxiv.org/abs/2203.00671}}
}

@inproceedings{LS20,
  title={Faster energy maximization for faster maximum flow},
  author={Liu, Yang P and Sidford, Aaron},
  booktitle={Proceedings of the 52nd Annual ACM SIGACT Symposium on Theory of Computing},
  pages={803--814},
  year={2020}
}

@inproceedings{DGGLPSY22,
author = {Sally Dong and Yu Gao and Gramoz Goranci and Yin Tat Lee and Richard Peng and Sushant Sachdeva and Guanghao Ye},
title = {Nested Dissection Meets IPMs: Planar Min-Cost Flow in Nearly-Linear Time},
booktitle = {Proceedings of the 2022 Annual ACM-SIAM Symposium on Discrete Algorithms (SODA)},
chapter = {},
pages = {124-153},
doi = {10.1137/1.9781611977073.7},
URL = {https://epubs.siam.org/doi/abs/10.1137/1.9781611977073.7},
eprint = {https://epubs.siam.org/doi/pdf/10.1137/1.9781611977073.7},
year = {2022}
}

@inproceedings{AS20,
  title={Faster p-norm minimizing flows, via smoothed q-norm problems},
  author={Adil, Deeksha and Sachdeva, Sushant},
  booktitle={Proceedings of the Fourteenth Annual ACM-SIAM Symposium on Discrete Algorithms},
  pages={892--910},
  year={2020},
  organization={SIAM}
}

@inproceedings{AMV20,
  title={Circulation control for faster minimum cost flow in unit-capacity graphs},
  author={Axiotis, Kyriakos and M{\k{a}}dry, Aleksander and Vladu, Adrian},
  booktitle={2020 IEEE 61st Annual Symposium on Foundations of Computer Science (FOCS)},
  pages={93--104},
  year={2020},
  organization={IEEE}
}

@inproceedings{KMP11,
  title={A nearly-m log n time solver for sdd linear systems},
  author={Koutis, Ioannis and Miller, Gary L and Peng, Richard},
  booktitle={2011 IEEE 52nd Annual Symposium on Foundations of Computer Science},
  pages={590--598},
  year={2011},
  organization={IEEE}
}

@article{APS19,
  title={Fast, provably convergent irls algorithm for p-norm linear regression},
  author={Adil, Deeksha and Peng, Richard and Sachdeva, Sushant},
  journal={Advances in Neural Information Processing Systems},
  volume={32},
  year={2019}
}

@inproceedings{GK21,
  title={Simple dynamic algorithms for maximal independent set, maximum flow and maximum matching},
  author={Gupta, Manoj and Khan, Shahbaz},
  booktitle={Symposium on Simplicity in Algorithms (SOSA)},
  pages={86--91},
  year={2021},
  organization={SIAM}
}

@inproceedings{GH23,
  title={Efficient Data Structures for Incremental Exact and Approximate Maximum Flow},
  author={Goranci, Gramoz and Henzinger, Monika},
  booktitle={50th International Colloquium on Automata, Languages, and Programming (ICALP 2023)},
  year={2023},
  organization={Schloss Dagstuhl-Leibniz-Zentrum f{\"u}r Informatik}
}

@inproceedings{BLS23,
  title={Dynamic Maxflow via Dynamic Interior Point Methods},
  author={van den Brand, Jan and Liu, Yang P and Sidford, Aaron},
  booktitle={Proceedings of the 55th Annual ACM Symposium on Theory of Computing},
  pages={1215--1228},
  year={2023}
}

@article{ABKS21,
  title={Almost-linear-time Weighted $\ell_p$-norm Solvers in Slightly Dense Graphs via Sparsification},
  author={Adil, Deeksha and Bullins, Brian and Kyng, Rasmus and Sachdeva, Sushant},
  journal={arXiv preprint arXiv:2102.06977},
  year={2021}
}

@inproceedings{B23,
  title = {Dynamic Algorithms for Maximum Matching Size},
  booktitle = {Proceedings of the 2023 {{Annual ACM-SIAM Symposium}} on {{Discrete Algorithms}} ({{SODA}})},
  author = {Behnezhad, Soheil},
  year = {2023},
  pages = {129--162},
  publisher = {{SIAM}}
}

@inproceedings{BKSW23,
  title = {Dynamic Matching with Better-than-2 Approximation in Polylogarithmic Update Time},
  booktitle = {Proceedings of the 2023 {{Annual ACM-SIAM Symposium}} on {{Discrete Algorithms}} ({{SODA}})},
  author = {Bhattacharya, Sayan and Kiss, Peter and Saranurak, Thatchaphol and Wajc, David},
  year = {2023},
  pages = {100--128},
  publisher = {{SIAM}}
}

@inproceedings{JS21,
  title={Ultrasparse ultrasparsifiers and faster laplacian system solvers},
  author={Jambulapati, Arun and Sidford, Aaron},
  booktitle={Proceedings of the 2021 ACM-SIAM Symposium on Discrete Algorithms (SODA)},
  pages={540--559},
  year={2021},
  organization={SIAM}
}

@inproceedings{CKMPPRX14,
  title={Solving SDD linear systems in nearly m log1/2 n time},
  author={Cohen, Michael B and Kyng, Rasmus and Miller, Gary L and Pachocki, Jakub W and Peng, Richard and Rao, Anup B and Xu, Shen Chen},
  booktitle={Proceedings of the forty-sixth annual ACM symposium on Theory of computing},
  pages={343--352},
  year={2014}
}

@inproceedings{SS08,
  title={Graph sparsification by effective resistances},
  author={Spielman, Daniel A and Srivastava, Nikhil},
  booktitle={Proceedings of the fortieth annual ACM symposium on Theory of computing},
  pages={563--568},
  year={2008}
}

@inproceedings{detMaxFlow,
  title={A Deterministic Almost-Linear Time Algorithm for Minimum-Cost Flow},
  author={Brand, Jan van den  and Chen, Li and Kyng, Rasmus and Liu, Yang P and Peng, Richard and Gutenberg, Maximilian Probst and Sachdeva, Sushant and Sidford, Aaron},
  booktitle={2022 IEEE 64rd Annual Symposium on Foundations of Computer Science (FOCS)},
  year={2023},
  organization={IEEE}
}
